\documentclass[preprint,11pt,xcolor=dvipsnames, a4wide]{article}

\usepackage{amsthm,amsmath,natbib, mathtools}
\usepackage{ifthen}
\usepackage{enumitem}
\usepackage{amssymb,dsfont,url,color,booktabs}
\usepackage[x11names]{xcolor}
 \usepackage{tikz}
 
 \usepackage{enumitem}
 \usepackage{geometry}
\usepackage{rotating}
\usepackage{epstopdf}
\usepackage{placeins} 		
\usepackage{algorithm}
\usepackage[noend]{algpseudocode}
\usepackage{chngcntr}
\usepackage{bbm}
\usepackage{array}
\usepackage[export]{adjustbox}[2011/08/13] 
\usepackage{caption}
\usepackage{multirow}
\usepackage{todonotes}

\usepackage[hang]{subfigure} 
\usepackage{wrapfig} 
\usepackage{tikz}
\usetikzlibrary{arrows, patterns}
\tikzset{
  schraffiert/.style={pattern=horizontal lines,pattern color=#1},
  schraffiert/.default=black
}
\usepackage{todonotes}

\newcommand*{\IE}{\mathbb{E}}
\newcommand*{\IR}{\mathbb{R}}

\newcommand*{\IC}{\mathbb{C}}
\newcommand*{\IN}{\mathbb{N}}

\renewcommand*\d{\mathop{}\!\mathrm{d}}

\newcommand*{\conv}{\ast} 

\newcommand{\sgn}{\mathop{\mathrm{sgn}}}
\newcommand{\norm}[1]{{\left\Vert #1 \right\Vert}}

\theoremstyle{plain}
\newtheorem{definition}{Definition}[section]
\newtheorem{theorem}{Theorem}[section]
\newtheorem{lemma}[theorem]{Lemma}

\newtheorem{corollary}[theorem]{Corollary}
\newtheorem{defin}[theorem]{Definition}
\newtheorem{remark}[theorem]{Remark}
\newtheorem{example}[theorem]{Example}

\newcommand{\dd}[1]{\operatorname{d}\!#1}

\newcommand{\ee}[1]{\operatorname{e}^{#1}}

\newcommand{\OF}{\mathcal{F}}

\newcommand{\notiz}[1]{\relax}

\newcommand{\zitep}[1]{\relax}
\newcommand{\skr}{\rangle}
\newcommand{\1}{\mathds 1}            

\newcommand{\nn}{\mathds N}
\newcommand{\N}{\mathds N}
 
\newcommand{\OA}{\mathcal{A\,}}

\newcommand{\rr}{\mathds R}

\newcommand{\rrd}{\mathds{R}^d}

\newcommand{\skl}{\langle}

\newcommand{\Dt}{\Delta t}
\newcommand*{\lHopital}{\text{l'H\^{o}pital}}

\newcommand{\trans}{\top}

\newcommand{\Price}[1][]{
		\ifthenelse{\equal{#1}{}}{\mathit{Price}}{\Price{}^{#1}}
	} 

\usepackage{mathrsfs}

\newcommand{\spann}{\operatorname{span}}

\newcommand{\tild}{~}

\newlength{\wordlength}

\begin{document}





\title{\textbf{ A Flexible Galerkin Scheme for Option Pricing in  L\'evy Models}}

\author{Maximilian Ga{\ss} and Kathrin Glau\\\\
}
\date{\today\\\indent Technische Universit{\"a}t M{\"u}nchen, Center for Mathematics\\ \indent kathrin.glau@tum.de, maximilian.gass@mytum.de}
\maketitle



\begin{abstract}
One popular approach to option pricing in L{\'e}vy models is through solving the related partial integro differential equation (PIDE). For the numerical solution of such equations powerful Galerkin methods have been put forward e.g. by~\cite{HilberReichmannSchwabWinter2013}. As in practice large classes of models are maintained simultaneously, flexibility in the driving L{\'e}vy model is crucial for the implementation of these powerful tools. In this article we provide such a flexible finite element Galerkin method. To this end we exploit the Fourier representation of the infinitesimal generator, i.e.\ the related symbol, which is explicitly available for the most relevant L{\'e}vy models. Empirical studies for the Merton, NIG and CGMY model confirm the numerical feasibility of the method.
\end{abstract}
%



\section{Introduction}

In computational finance, methods to solve partial differential equations come into play, when both run time and accuracy matters. In contrast to Monte Carlo for example, run-time is very appealing and a deterministic and conservative error analysis is established and well understood. 
And, compared to Fourier methods, the possibility to capture path-dependent features like early exercise and barriers is naturally built in. 
Within these appealing features lies the capacity to attract interest from academia and satisfy  the needs of the financial industry alike.

In academia, a series of publications by \cite{ContVoltchkova2005a}, Hilber, Reich, Schwab and Winter (2009)\nocite{HilberReichSchwabWinter2009}, Salmi, Toivanen and Sydow (2014)\nocite{SalmiToivanenSydow2014}, \cite{Itkin2015} and \cite{Glau2016} and the monograph of Hilber, Reichmann, Schwab and Winter (2013)\nocite{HilberReichmannSchwabWinter2013} have opened the theory to include even more sophisticated models of L\'evy type, resulting in Partial \textit{Integro} Differential Equations (PIDEs). The theoretical results have been validated by sophisticated numerical studies. In this context, Schwab and his working group in particular have taken the lead and unveiled the potential of PIDE theory for practical purposes in the financial industry. Combining state of the art compression techniques with a wavelet basis finite element setup has resulted in a numerical framework for option pricing in advanced and multivariate jump models and thereby moved academic boundaries.

In the financial industry an awareness of the full potential of these tools is yet to be developed. Advocating the advancement of numerical methods one must acknowledge what practice cherishes most in them. Due to model uncertainty and behavioural characteristics of different portfolios, financial institutions need to deal with a number of different pricing models in parallel. Or, in the words of \cite{FoellmerDMV}: "In any case, the signal towards the practitioners of risk management is clear: do not commit yourself to a single model, remain flexible, vary
the models in accordance with the problem at hand, always keeping in mind the worst case scenario."\footnote{Tranlsated from German.}
 These features need to be reflected in the numerical environment. 

In this article, we aim at reconciling capacities of state of the art P(I)DE 
tools with the flexibility regarding different model choices as required by industry. Desirable features that such an implementation must offer include
	\begin{enumerate}
		\item[(1)] a degree of accuracy that reaches levels relevant to practical applications and that can be measured and controlled by a theoretical error analysis,
		\item[(2)] fast run times, 
		\item[(3)] low and feasible implementational and maintenance cost,
		\item[(4)] a flexibility of the toolbox towards different options and models.
	\end{enumerate}
Two standard methods are available for solving P(I)DEs, that is the finite difference approach and the finite element method. More recently, also radial basis methods have been pushed forward to solve pricing PIDEs.
In principle these concepts can be implemented in such a way that they achieve the desired features 1.--3. and implementations for a variety of models and option types have already been developed: Finite difference schemes solving PIDEs for pricing European and barrier options with an implementation for Merton and Variance Gamma are provided in \cite{ContVoltchkova2005a} and \cite{ContVoltchkova2005b}. 
The method has been further developed in different directions, we mention one example,
 \cite{ItkinCarr2012}, who use a special representation of the equation to derive a finite different scheme for jump diffusions with jump intensity of tempered stable type.
Wavelet-Galerkin methods for PIDEs related to a class generalizing tempered stable L\'evy processes are derived in  \cite{MatacheNitscheSchwab2005} for American options and see e.g.\ \cite{MarazzinaReichmannSchwab2012} for a high-dimensional extension. Radial basis for the Merton and Kou model, American and European options are provided by \cite{ChanHubbert2014} and further developed for CGMY models by \cite{BrummelhuisChan2014}.

An implementation that is flexible in the driving model as well as in the option type first of all requires a problem formulation covering the collectivity of envisaged models and options. In view of feature (1), a unified approach to the error analysis of the resulting schemes is of equal importance. Galerkin methods, accruing from the Hilbert space formulation of the Kolmogorov equation, seem to be predestined to deliver the adequate level of abstraction for this task. It is this abstract level that makes Galerkin methods flexible in the option types and the dimensions of the underlying driving process. 
Consequently, even though Galerkin methods seems more involved at first glance in comparison to finite difference schemes, they still promise to lead to a more lucid code that is easier to maintain and to extend. Another fundamental advantage of Galerkin methods is its theoretical framework that allows for clear and extensive convergence analysis and error estimates, which is of great importance for controlling methodological risk in finance. 
The finite element, or more general Galerkin methods are therefore our methods of choice.


%
%

Unfortunately, even the finite element methods faces numerical challenges when implementing L\'evy model based pricing tools. More precisely, the L\'evy operator that determines the stiffness matrix is of integro differential type. First, the resulting matrix is densely populated and in general not symmetric. Second, and even more severe, the matrix entries typically are not explicitly available. Instead, they require the evaluation of double integral terms possibly  involving a numerically inaccessible L\'evy measure. In these cases, a thorough analysis of the respective integrals may lead to approximation schemes deriving the stiffness matrix entries with the required precision. Pursuing this way, however, most likely results in a model specific scheme, contradicting requirement (4).

In this article, our aim is \emph{the development of a model-independent approach to set up a FEM solver for option pricing in L\'evy models} that we call \emph{symbol method}. We address this goal by expressing the operator in the Fourier space. This means accessing the model specific information via the symbol. In contrast to the operator, the symbol is explicitly available for a variety of models and is thus numerically accessible. Further advantages will be highlighted in subsequent sections.

Section \ref{sec-framework} introduces the theoretical framework for our PIDEs of interest and their weak formulation. The next section describes the solution scheme, that is the Galerkin approximation in space. We investigate the scheme  with regard to the numerical challenges arising during its implementation. Section \ref{sec-symbolmethod} introduces the symbol method itself. All components of the FEM solver are expressed in Fourier space. The subsequent numerical evaluation of the stiffness matrix entries is supported by an elementary approximation result. Several examples of symbols for well-known L\'evy models confirm the wide applicability of the method and its numerical advantages. Two proposals for the implementation of basis functions are presented. The numerical studies in Section \ref{sec-num} confirm theoretical prescribed rates of convergence and validate the claim of numerical feasibility.

%
%
%
%
%
%

\section{Kolmogorov equations for option pricing in L\'evy models}\label{sec-framework}
We first introduce the underlying stochastic processes, the Kolmogorov equation, its weak formulation as well as the solution spaces of our choice. 
\subsection{L\'evy processes} 

Let a stochastic basis $(\Omega,\OF, (\OF_t)_{0\le t\le T}, P)$ be given and let
$L$ be an $\rrd$-valued \emph{L\'evy process} with characteristics $(b,\sigma,F;h)$, i.e.\ for fixed $t\ge0$ its characteristic function is given by
\begin{align}\label{eq-charPIIAC}
E\ee{i \skl \xi, L_t \skr} = \ee{ -t A(-i\xi)\dd s}\quad \text{for every }\xi\in\rrd,
\end{align}
where the \emph{symbol of the process }is defined as
\begin{equation}\label{def-A}
A(\xi) := \frac{1}{2}\langle \xi,\sigma \xi\rangle + i\langle \xi,b\rangle 
- \int_{\rrd}\left(\ee{-i\langle \xi,y\rangle} -1+ i\langle \xi,h(y)\rangle\right)\,F(\dd y).
\end{equation}
Here, $\sigma$ is a symmetric, positive semi-definite $d\times d$-matrix, $b\in \rr^d$, and $F$ is a L\'evy measure, i.e. a positive Borel measure on $\rrd$ with $F(\{0\})=0$ and $\int_{\rr^d} (|x|^2 \wedge 1) F(\dd x)  < \infty$. Moreover, $h$ is a truncation function i.e. $h:\rr^d\to\rr$ such that $\int_{\{|x|>1\}} h(x) F(\dd x)<\infty$ with $h(x)=x$ in a neighbourhood of $0$.
The  \emph{Kolmogorov operator of a L\'evy process} $L$ with characteristics $(b,\sigma,F;h)$ is given by
\begin{align}\label{def-opA}
\begin{split}
\OA \varphi(x)\coloneqq & - \frac{1}{2}\sum_{j,k=1}^d \sigma^{j,k}\frac{\partial^2 \varphi}{\partial x_j\partial x_k}(x)-\sum_{j=1}^d b^j\frac{\partial \varphi}{\partial x_j}(x)\\
&-\int_{\rr^d}\Big( \varphi(x+y)-\varphi(x)-\sum_{j=1}^d\frac{\partial \varphi}{\partial x_j}(x) \,  h_j(y)\Big)F(\dd y)
\end{split}
\end{align}
for every $\varphi\in C^\infty_0(\rrd)$, where $h_j$ denotes the $j$-th component of the truncation function $h$.

\subsection{Kolmogorov equation in variational form}
Key for the variational formulation of
Kolmogorov equation
\begin{align}\label{parabolic-eq-origin}
\partial_t u + \OA u &= f\\
u(0)&=g\label{parabolic-eq-initial}
\end{align}
is the definition of the bilinear form 
\begin{equation}\label{def-aCinfty0}
a(\varphi,\psi) := \int_{\rrd} (\OA \varphi)(x)\psi(x)\dd x\qquad \text{for all }\varphi,\psi\in C^\infty_0(\rrd).
\end{equation}
It is one of the major advantages of variational formulations of evolution equations that solution spaces of low regularity, as compared to the space $C^2$ for example, are incorporated in an elegant way.
Departing from the space $C^\infty_0(\rrd)$ of smooth functions with compact support, we can select from a large variety of function spaces $V$ that is characterized by the following assumption.
\begin{enumerate}
\item[(A1)] $V$ and $H$ are Hilbert spaces such that $C^\infty_0(\rrd)$ lies dense in $V$ and there exists a continuous embedding from $V$ into $H$. 
\end{enumerate}
Existence and uniqueness of a variational solution critically hinges on the following two properties of the bilinear form:
\begin{enumerate}
\item[(A2)]\textit{Continuity}: There exists a constant $C>0$ such that
\begin{equation*}
\big|a(\varphi,\psi) \big|\le C \|\varphi\|_V\|\psi\|_V\qquad \text{for all }\varphi,\psi\in C^\infty_0(\rrd).
\end{equation*}
\item[(A3)]\textit{G{\aa}rding inequality}: There exists constants $G>0$ and $G'\ge0$ such that
\begin{equation*}
a(\varphi,\varphi) \ge G\|\varphi\|_V^2 - G'\|\varphi\|_H^2\qquad \text{for all }\varphi\in C^\infty_0(\rrd).
\end{equation*}
\end{enumerate}
We observe that due to (A1) and (A2), the bilinear form $a$ possesses a unique continuous bilinear extension $a:V\times V$ that is continuous, i.e. for a constant $C>0$ we have  $\big|a(\varphi,\psi) \big|\le C \|\varphi\|_V\|\psi\|_V$ for all $v\in V$. Also (A3) holds for all $v\in V$.

As $V$ is separable, this is also true for $H$ and one can find  a continuous embedding from $H$ to the dual space $V^\ast$ of $V$, i.e. $(V,H,V^\ast)$ is a Gelfand triplet. We then denote by $L^2\big(0,T; H\big)$ the space of weakly measurable functions $u:[0,T]\to H$ with $\int_0^T\|u(t)\|_H^2 \dd t < \infty$ and by $\partial_t u$ the derivative of $u$  with respect to time in the distributional sense. For a detailed definition which relies on the Bochner integral we refer to Section 24.2 in \cite{Wloka-english}.
The Sobolev space 
\begin{equation}\label{def-W1}
W^1( 0,T; V,H) := \Big\{ u\in L^2\big(0,T;V\big) \,\Big| \,\partial_t u\in L^2\big(0,T; V^\ast\big) \Big\},
\end{equation}
will play the role of the solution space in the variational formulation of Kolmogorov equation \eqref{parabolic-eq-origin}, \eqref{parabolic-eq-initial}.

\begin{defin}
Let $f\in L^2\big(0,T; V^\ast\big)$ and $g\in H$. Then 
$u\in W^1( 0,T; V,H)$ is a \emph{variational solution} of Kolmogorov equation \eqref{parabolic-eq-origin}, if for almost every $t\in(0,T)$, 
\begin{equation}\label{def-para}
 \skl \partial_t u(t), v\skr_{H} + a( u(t), v)  =\, \skl f(t) | v\skr_{V^\ast\times V}\quad \text{for all }v\in V
\end{equation}
and $u(t)$ converges to $g$ for $t\downarrow0$ in the norm of $H$.
\end{defin}
\begin{remark}
Assumptions (A1)--(A3) guarantee the existence and uniqueness of a variational solution $u\in W^1( 0,T; V,H)$ of \eqref{def-para},
see for instance 
 Theorem~23.A in \cite{Zeidler}.
\end{remark}

\subsection{Solution spaces}
Definition \ref{def-aCinfty0} is based on the $L^2$-scalar product and is appropriate for variational equations in Sobolev spaces. Then, typically $H=L^2$. For Kolmogorov equations for option prices the initial condition $g$ in \eqref{parabolic-eq-initial} plays the role of the (logarithmically transformed) payoff function of the option. For a call option with strike $K$ it is of the form $x\mapsto(S_0\ee{x} - K)^+$, for a digital up and out option  it is given by $x\mapsto\1_{x<b}$ with for some $b\in\rr$.
We thus have to observe that the initial condition $g$ is not square integrable for most of the typical cases of interest. 
Therefore, we base our analysis more generally on exponentially weighted $L^2$ spaces: For $\eta\in\rrd$ let
\begin{equation*}
L^2_\eta(\rrd) := \big\{u\in L^1_{loc}(\rrd)\,|\, u \ee{\skl\eta,\cdot\skr}\in L^2(\rrd)\big\},\quad \|u\|_{L^2_\eta}:=\bigg(\int_{\rrd} \big|u(x)\big|^2\ee{2\skl\eta,x\skr}\dd x\bigg)^{1/2}
\end{equation*}
and
\begin{equation}\label{def-aCinfty0-eta}
a(\varphi,\psi) := \skl \OA\varphi,\psi\skr_{L^2_\eta}=\int_{\rrd} (\OA \varphi)(x)\psi(x) \ee{2\skl\eta,x\skr}\dd x\quad \text{for all }\varphi,\psi\in C^\infty_0(\rrd).
\end{equation}
We notice that all assertions of the precedent section, concerning assumptions (A)--(A3) and variational equations hold for bilinear form $a$ defined by \eqref{def-aCinfty0-eta} instead of $a$ from \eqref{def-aCinfty0} as well.

As solution spaces $V$ we consider weighted Sobolev-Slobodeckii spaces. These have proven to apply to a large set of option types and models. We refer to \cite{EberleinGlau2013} and \cite{Glau2016},
where particularly Feynman-Kac type formulas have been derived linking European and path-dependent options to weak solutions of Kolmogorov equations in Sobolev-Slobodeckii spaces.

To introduce the spaces, we denote by $C_0^\infty(\rrd)$ the set of smooth real-valued functions with compact support in $\rrd$ and let
\begin{equation}\label{def_FT}
\OF(\varphi)(\xi):= \int_{\rrd}\ee{i\skl \xi,x\skr} \varphi(x) \dd x
\end{equation}
be the Fourier transform of $\varphi\in C^\infty_0(\rrd)$ and $\OF^{-1}$ be its inverse. 
We define the \emph{exponentially weighted Sobolev-Slobodeckii space} $H^\alpha_\eta(\rrd)$ with index $\alpha\ge0$ and weight $\eta\in \rrd$ as the completion of $C_0^\infty(\rrd)$ with respect to the norm $\|\cdot\|_{H^\alpha_\eta}$ given by
\begin{equation}\label{def_normHseta}
\|\varphi\|_{H^\alpha_\eta}^2:= \int_{\rrd} \big(1+|\xi|\big)^{2\alpha}\big|\OF(\varphi)(\xi - i \eta)\big|^2 \dd \xi.
\end{equation}
Furthermore, we denote the dual space of $H^\alpha_\eta(\rrd)$ by $\big(H^\alpha_\eta(\rrd))^\ast$.

\section{\textbf{Implementational Challenges}}

\subsection{Abstract Galerkin approximation in space}

For a countable Riesz basis $\{\varphi_1,\varphi_2,\ldots\}$ of $V$ we define
\begin{equation*}
 V_N:=\operatorname{span}\{\varphi_1,\ldots,\varphi_N\}\qquad\text{for all }N\in\N.
\end{equation*}
Since $V$ is dense in $H$, we may further choose $g_N$ in\tild $V_N$ such that 
$g_N \rightarrow u(0)
$ in\tild$H$.
For each fixed $N\in\nn$ the semidiscrete problem is defined by restricting \eqref{def-para} to the finite dimensional space:  
\textit{
Find a function $v_N \in  W^1(0,T; V_N; H\cap V_N )$ that satisfies for all $\chi \in C^\infty_0(0,T)$ and $\varphi \in V_N$,
\begin{equation}\label{gl-variation_V_n}
\begin{split}
-\int_0^T\!\! \skl v_N(t), \varphi\skr_{L^{2}} \,\dot{\chi}(t) \dd t + \int_0^T\!\! a\big(v_N(t), \varphi\big) \,\chi(t) \dd t 
&=
\int_0^T \skl f(t)|\varphi\skr_{V^\ast\times V} \, \chi(t) \dd t
 \\
v_N(0) &= g_N.
\end{split}
\end{equation}}
As a result of the elegant Hilbert space formulation, the semidiscrete problem \eqref{gl-variation_V_n} is uniquely solvable and the convergence of the sequence $v_N$ to $v$ is guaranteed, see Theorem\tild 23.A. and Remark 23.25 in \cite{Zeidler}.

The major advantage of equation \eqref{gl-variation_V_n} in regard to implementation is that 
 it suffices to insert the basis functions as test functions.
Thus, denoting $g_N = \sum_{k=1}^N \alpha_{k} \varphi_k$
and $v_N(t)\coloneqq \sum_{k=1}^N V_k (t) w_k$ we arrive at
\begin{align*}
\sum_{k=1}^N \dot{V}_{k}(t) \skl \varphi_k,\varphi_j\skr_{L^{2}} + \sum_{k=1}^N V_{k}(t) a\big(\varphi_k,\varphi_j \big) 
&= \skl f(t)|\varphi_j\skr_{V^\ast\times V}\\
V_k(0)&= \alpha_k\quad\text{for all }k=1,\ldots,N.
\end{align*}
Written in matrix form the problem is to find $V:[0,T] \rightarrow \rr^{N}$ such that
\begin{align}\label{Pide-matrixform}
M \dot{V}(t) + A V(t) &= F(t) \\
V(0) &=\alpha,\label{Pide-matrixform-initial}
\end{align}
where $F = (F_1, \ldots, F_N)^\trans$ with $F_k(t) = \skl f(t)|\varphi_k\skr_{V^\ast\times V}$ for $k=1,\ldots,N$, $\alpha =(\alpha_1,\ldots,\alpha_N)^\trans$, and the \emph{mass matrix} $M$ and \textit{stiffness matrix} $A$ are given\tild by
\begin{equation}
M_{jk} =  \skl \varphi_k,\varphi_j\skr_{L^{2}} ,\qquad A_{jk} = a\big(\varphi_k,\varphi_j \big)\qquad \text{for all }j,k=1,\ldots,N.
\end{equation}


\subsection{Flexible implementation for different driving L\'evy processes}
\label{sec:FlexibleImplementation}
We inspect equation \eqref{Pide-matrixform-initial} in regard to 
flexibility towards different options as well as models. All ingredients in \eqref{Pide-matrixform} depend on the choice of the basis. While $M$ is independent of the specific problem at hand, $F$ and $\alpha$ represent the input data and therefore may vary for different option types. The stiffness matrix $A$ that carries the information of the driving driving process. So in order to obtain flexibility towards model types, we need a generic way to compute the entries of the stiffness matrix. For smooth basis functions with compact support and solution spaces without weighting, i.e.\ $\eta=0$, according to \eqref{def-A}, \eqref{def-aCinfty0}, the stiffness matrix entries are given  by
\begin{align}
a\big(\varphi_k,\varphi_j \big) 
= &-\sum_{l,m=1}^d  \frac{\sigma^{j,m}}{2}\int_{\rrd} \big( \partial_l\partial_m \varphi_k(x)\big) \varphi_j(x)\dd x
-\sum_{l=1}^d b^l\int_{\rrd} \partial_l \varphi_k(x)\varphi_j(x)\dd x\nonumber\\
&-\int_{\rr^d}\int_{\rr^d}\Big( \varphi_k(x+y)-\varphi_k(x)-\sum_{l=1}^d\partial_l \varphi_k(x) \,  h_l(y)\Big)F(\dd y)\varphi_j(x)\dd x.\label{eq-ajump}
\end{align}
Typical basis functions are not smooth. Therefore it is not a priori clear if the integral representation \eqref{eq-ajump} extends to the usual basis functions. Observe hat an extension of this representation requires some care: For a large and important class of pure jump L\'evy processes, the solution spaces are Sobolev-Slobodeckii spaces of fractional order, i.e.\
$H^\alpha$ with some $0<\alpha<1$. For functions in $H^\alpha$ with $\alpha<1$, however, the first order weak derivative in \eqref{eq-ajump} is not defined and therewith this integral representation of the bilinear form is not well-defined. Seen that the basis functions usually are in $H^1$, also in the more challenging case of solution spaces with fractional order derivatives we derive the validity of the representation under appropriate assumptions.

\begin{lemma}\label{lem-a=forH1}
Let $d=1$.
Let $a$ be defined by \eqref{def-aCinfty0-eta}. Assume (A1)--(A3) for $a$, $V$ and $H$ and denote by $a:V\times V$ its unique bilinear continuous extension. If $H^1_\eta(\rr)\subset V$, we have for every $\varphi,\psi\in H^1_\eta(\rr)$,
\begin{align}
a(\varphi,\psi ) 
= &\frac{\sigma}{2} \int_{\rr} \varphi'(x)\psi'(x) \ee{2 \eta x} \dd x
- b(\eta, \sigma,F) \int_{\rr} \varphi'(x)\psi(x) \ee{2\eta x}\dd x\nonumber\\
&-\int_{\rr}\int_{|y|<1} \int_0^y\int_0^z\varphi'(x+v)\dd v \dd zF(\dd y)\big(\psi'(x) + 2\eta \psi(x) \big)\ee{2\skl\eta,x\skr}\dd x
\label{eq-aH1}
\\
&-\int_{\rr}\int_{|y|>1}\big( \varphi(x+y)-\varphi(x)\big)F(\dd y)\psi(x) \ee{2\skl\eta,x\skr}\dd x\nonumber
\end{align}
with
\begin{equation*}
b(\eta, \sigma,F) =  b - 2 \sigma \eta + \int_{|y|<1}\big( y-h(y)\big)F(\dd y) - \int_{|y|>1}h(y)F(\dd y).
\end{equation*}
\end{lemma}
The proof is provided in Section \ref{sec-proof-lem-a=}.

%
%

Inspecting the expression for the bilinear form we encounter several numerical challenges due to the integral part---stemming from the jumps of the process:
\begin{enumerate}
\item 
The appealing tridiagonal structure of the stiffness matrix related to the Black-Scholes equation does not extend to the general L\'evy setting. Instead, the stiffness matrix is densely populated. Pleasantly, it is still a Toeplitz matrix.

\item 
For some choices of L\'evy measures and bases the stiffness matrix entries may be derived in closed form. This is for instance the case for the Merton model and piecewise linear basis functions. Following Section 10.6.2 in \cite{HilberReichmannSchwabWinter2013}, the stiffness matrix entries may be derived in semi-closed form expressions for a further group of jump intensities including tempered stable, CGMY and KoBoL processes and the choice of piecewise linear basis functions.
\end{enumerate}

%

An implementation that is flexible in the driving L\'evy process therefore has to rely on numerical approximations of the entries of the stiffness matrix. These approximations inevitably affect the accuracy of the solution to the scheme \eqref{Pide-matrixform}--\eqref{Pide-matrixform-initial}. The question rises: 
 \textit{How accurate does the integration routine have to be chosen in order to meet a desired accuracy of the solution $V$?}
In order to gain a first practical insight in the magnitude of the error resulting from an inaccuracy in the stiffness matrix entries, we next conduct an empirical investigation.


\subsection{An accuracy study for approximations of the stiffness matrix}
\label{sec:AccuracyStudy}

To this end we conduct an accuracy study for the stiffness matrix entries simulating the propagation of integration errors in the stiffness matrix up until computed option prices. We choose the Black-Scholes model for which the precise values of all matrix entries are known and assume a market volatility of $\sigma=0.2$ and interest rate $r=0.01$ for this study. With a classic FEM solver we price a put option with strike $K=1$ and maturities $T\in[0,3]$ for current values of the stock $S_0\in [S_\text{min}, S_\text{max}]$ with $S_\text{min}=0.01$ and $S_\text{max} = 10$. We set the number of involved FEM hat functions to $N=150$, resulting in a mesh with $150$ inner grid nodes and mesh fineness $h=0.0464$. We know the mass matrix of the Black-Scholes model to be given by
	\begin{align*}
		M &= \frac{h}{6} \begin{pmatrix}
												4 & 1 & 0 & \cdots & 0 \\
												1 & 4 & 1 & \ddots & \vdots \\
												0 & \ddots & \ddots & \ddots & 0 \\
												\vdots & \ddots & 1 & 4 & 1\\
												0 & \cdots & 0 & 1 & 4\\
 											\end{pmatrix}\in\IR^{N\times N},
	\end{align*}
and the stiffness matrix to be given by
	\begin{equation}
	\label{eq:FEMdistStiffnessOrg}
	 		A = A^\text{bs} = A^{(1)} + A^{(2)} + rM \in\IR^{N\times N}, 
	\end{equation}
where
	\begin{equation*}
			A^{(1)}\!=\!\frac{1}{2}\!\left(\!r-\frac{\sigma^2}{2}\right)\!\begin{pmatrix}
												0 & -1 & 0 & \cdots & 0 \\
												1 & 0 & -1 & \ddots & \vdots \\
												0 & \ddots & \ddots & \ddots & 0 \\
												\vdots & \ddots & 1 & 0 & -1\\
												0 & \cdots & 0 & 1 & 0\\
 											\end{pmatrix},\quad
 			A^{(2)}\!=\!\frac{\sigma^2}{2}\frac{1}{h} \begin{pmatrix}
												2 & -1 & 0 & \cdots & 0 \\
												-1 & 2 & -1 & \ddots & \vdots \\
												0 & \ddots & \ddots & \ddots & 0 \\
												\vdots & \ddots & -1 & 2 & -1\\
												0 & \cdots & 0 & -1 & 2\\
 											\end{pmatrix}.
	\end{equation*}
With these matrices we set up a theta scheme, $\theta = 0.5$, and derive Black-Scholes put option prices. 
We use these matrices to simulate how the resulting pricing surface is affected by inaccuracies that might occur when these integrals are solved numerically, instead. To this extent we take the correct stiffness matrix given by~\eqref{eq:FEMdistStiffnessOrg} and distort each of its entries randomly at different positions $D\in\IN$ after the decimal point by adding $\varepsilon_i^D = 10^{-(D-1)}\varepsilon_i$ with random $\varepsilon_i\in(-1,1)$ for $i\in\{-(N-1),\dots,-1,0,1,\dots,(N-1)\}$ onto the (side) diagonal $i$ of Matrix $A$. Each individual (side) diagonal of the original stiffness matrix is thus affected evenly, keeping the Toeplitz structure of the matrix intact. Since the value of $A_{ij}$ is only determined by the value of $j-i$, this distortion mimics the influence that integration inaccuracies would have.

\begin{figure}[h!]
\begin{center}
\makebox[0pt]{\includegraphics[scale=1]{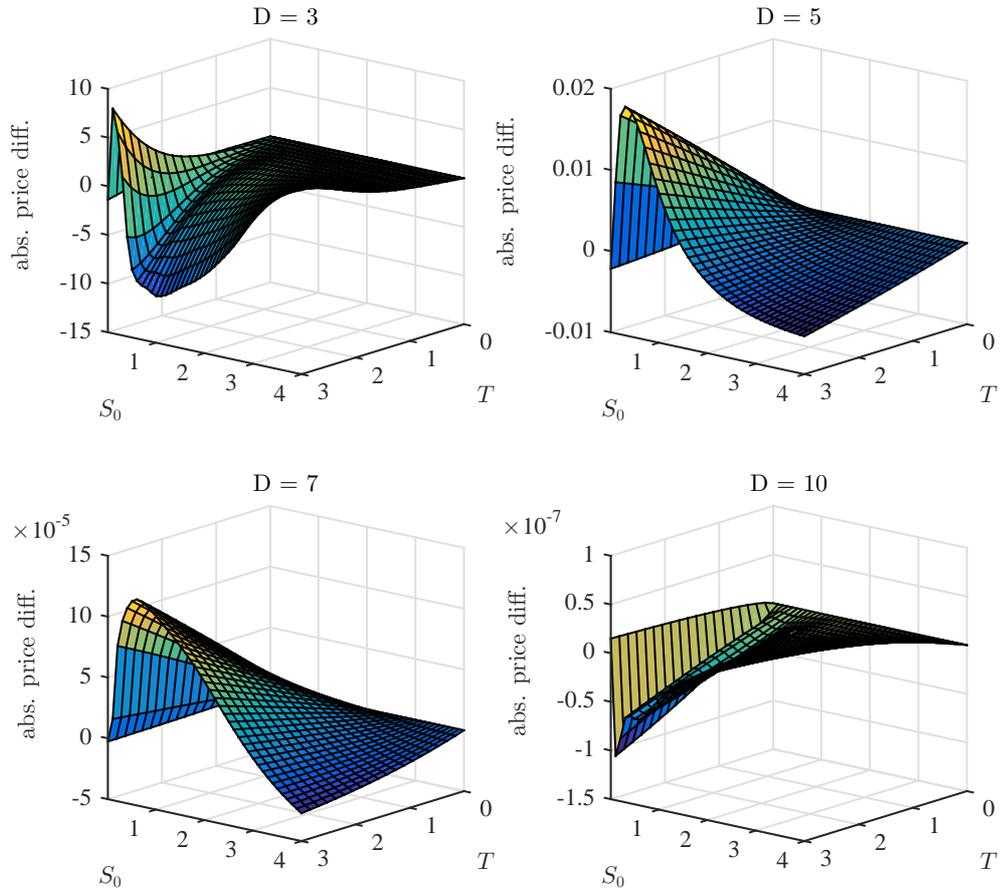}}
\caption{Absolute price differences resulting from a distortion of the stiffness matrix $A$. True and distorted prices describe the market value of a put option in the Black-Scholes model with the parametrization from Section \ref{sec:AccuracyStudy}. We compare the price surfaces coming from a theta scheme using the stiffness matrix $A$ given by~\eqref{eq:FEMdistStiffnessOrg} to the respective pricing surface when $A$ is replaced by $A^D_\text{distort}$, the distorted version of $A$ as defined in~\eqref{eq:FEMdistStiffness01}, for different values of $D\in\IN$. The influence of the distortion decreases in $D$.}
\label{fig:FEMOptionPriceDiffs4}
\end{center}
\end{figure}

Hence, for $D\in\IN$ we define the distorted stiffness matrix by
	\begin{equation}
	\label{eq:FEMdistStiffness01}
		A^D_\text{distort} = A + \varepsilon^D \in \IR^{N\times N},
	\end{equation}
with
	\begin{equation*}
		\varepsilon^D = 10^{-(D-1)}\begin{pmatrix}
						\varepsilon_{0} & \varepsilon_{1} & \varepsilon_{2} & \cdots & \cdots & \cdots & \varepsilon_{N-1} \\
						\varepsilon_{-1} & \varepsilon_{0} & \varepsilon_{1} & \ddots & \ddots & \ddots & \vdots \\
						\varepsilon_{-2} & \varepsilon_{-1} & \ddots & \ddots & \ddots &\ddots & \vdots \\
						\vdots & \ddots & \ddots & \ddots & \ddots & \ddots & \vdots \\						
						\vdots & \ddots & \ddots & \ddots & \ddots & \varepsilon_{1} & \varepsilon_{2}\\
						\vdots & \ddots & \ddots & \ddots & \varepsilon_{-1} & \varepsilon_{0} & \varepsilon_{1}\\			
						\varepsilon_{-(N-1)} & \cdots & \cdots & \cdots & \varepsilon_{-2} & \varepsilon_{-1} & \varepsilon_{0}\\
 					\end{pmatrix} \in \IR^{N\times N},
	\end{equation*}
with uniformly distributed $\varepsilon_i \in (-1,1)$, $i\in\{-(N-1),\dots,-1,0,1,\dots,(N-1)\}$, that are drawn independently from each other. Using these distorted stiffness matrices $A^D_\text{distort}$ for different values $D\in\IN$, we derive again price surfaces of the put option in the Black-Scholes model and compare the difference between the prices coming from the distorted stiffness matrix $A^D_\text{distort}\in\IR^{N\times N}$ to the prices from the intact stiffness matrix $A\in\IR^{N\times N}$.

Figure~\ref{fig:FEMOptionPriceDiffs4} shows the results of the study and emphasizes the necessity for high numerical accuracy in the computation of the stiffness matrix. We observe that the absolute price differences decrease almost linearly in $D$. An accuracy of $D=3$ corresponds to integration results that are exact up to the third digit after the decimal point. Pricing resulting from stiffness matrices computed with such a low integration accuracy are unacceptable. The respective pricing errors observable in the top left corner of Figure~\ref{fig:FEMOptionPriceDiffs4} indicate relative errors of several hundred percent points. With more precise integration results, the error decreases in $D$ until highly appealing pricing results are achieved for $D=7$ and beyond. The magnitude of the pricing error resulting from a distorted stiffness matrix emphasizes the necessity of being able to derive the stiffness matrix entries as accurately as possible. 

\FloatBarrier

\section{Fourier approach to the Kolmogorov equation}\label{sec-symbolmethod}

In regard to the high accuracy the approximation of the stiffness matrix entries needs to achieve, we would like to avoid numerical evaluations of the stiffness matrix entries on the basis of representation \eqref{eq-aH1}. Seeking for alternative representations of the stiffness matrix, let us point out that the symbol $A$ of the L\'evy process is always available. Even more, it is an explicit function of the parameter of the process and thus can be seen as the modelling quantity of the process as the examples \ref{ex:BSsymbol}--\ref{ex:NIGsymbol} below show. We therefore take a Fourier perspective on the variational formulation of the Kolmogorov equation. This is especially promising since the Kolmogorov operator $\OA$ of a L\'evy process is a pseudo differential operator with with symbol~$A$, 
\begin{equation}\label{eq-Aispseudo}
\OA \varphi=\OF^{-1}(A \OF(\varphi))\qquad\text{for all }\varphi\in C^\infty_0(\rrd),
\end{equation}
as elementary manipulations show.
Now Parseval's identity yields 
\begin{equation}
a(\varphi,\psi) = \frac{1}{(2\pi)^d}\int_{\rrd} \OF(\OA \varphi)(\xi)\overline{\OF(\psi)(\xi)}\dd \xi
\end{equation}
for all $\varphi,\psi\in C^\infty_0(\rrd)$, respectively,
\begin{equation}
a(\varphi,\psi) =\frac{1}{(2\pi)^d}\int_{\rrd}A(\xi) \OF(\varphi)(\xi)\overline{\OF(\psi)(\xi)}\dd\xi.
\end{equation}
This well-known identity has already proven highly beneficial for the analysis of the variational solutions of the Komogorov equations, compare e.g.\ \cite{HilberReichmannSchwabWinter2013}, \cite{Glau2013} and \cite{Glau2016}. Here we exploit this representation for the numerical implementation.

\begin{lemma}[Continuous extension of bilinear forms]
\label{lem:SymbolMethod}
Let $A$ be the symbol of a L{\'e}vy process given by the characteristic triplet $(b,\sigma,F)$. Denote by $\mathcal{A}:C_0^\infty(\IR^d,\IC)\rightarrow C^\infty(\IR^d,\IC)$ the pseudodifferential operator associated with symbol $A$. Furthermore, denote by \mbox{$a:C_0^\infty\times C_0^\infty\rightarrow\IC$} the bilinear form associated with the operator $\mathcal{A}$. Let $\eta\in\IR^d$. If
	\begin{enumerate}[label=\roman*)]
		\item the exponential moment condition
			\begin{equation}
				\int_{|x|>1}e^{-\langle \eta',x\rangle}F(\d{x}) < \infty
			\end{equation}
			holds for all $\eta'\in \sgn(\eta^1)[0,|\eta^1|]\times\cdots\times \sgn(\eta^d)[0,|\eta^d|]$ and
		\item there exists a constant $C_1>0$ with
			\begin{equation}
				|A(z)| \leq C_1(1+\norm{z})^\alpha
			\end{equation}
			for all $z\in U_{-\eta}$ where
			\begin{equation}
				U_{-\eta} = U_{-\eta^1} \times \cdots \times U_{-\eta^d}
			\end{equation}
			with $U_{-\eta^j} = \IR - i\sgn(\eta^j)[0,|\eta^j|)$,
	\end{enumerate}
then $a(\cdot,\cdot)$ possesses a unique linear extension $a:H_\eta^{\alpha/2}\times H_\eta^{\alpha/2}\rightarrow\IC$ which can be written as
	\begin{equation}
		a(\varphi,\psi) = \frac{1}{(2\pi)^d}\int_{\IR^d} A(\xi-i\eta)\widehat{\varphi}(\xi-i\eta)\overline{\widehat{\psi}(\xi-i\eta)}\d{\xi}
	\end{equation}
for all $\varphi,\psi\in H_\eta^{\alpha/2}(\IR^d)$.
\end{lemma}

\begin{proof}
The proof can be found in \cite{EberleinGlau2012} using Theorem~4.1 therein and Parseval's identity.
\end{proof}


In order to gain first insight in the convergence analysis, we fix a level $N$ in the Galerkin scheme and derive conditions for the convergence of the sequence of weak solutions that we obtain by approximating the stiffness matrix entries. In the implementation in Section \ref{sec-num}  below we will also approximately compute the right hand side of the equation. We therefore more generally consider sequences of stiffness matrices, right hand sides and initial conditions.

As usual, we denote for a given bilinear form $a:V\times V\to \rr$ the associated operator $\OA:V\to V^\ast$ defined by $\OA(u)(v):=a(u,v)$ for all $u,v\in V$.

\begin{lemma}\label{lem-conv1}
Let $V$, $H$ and $a:V\times V\to\rr$ satisfy (A1)--(A3).
Let $X:=\spann\{\varphi_1,\ldots,\varphi_N\}\subset V$ and for each $n\in\nn$ let
\begin{enumerate}[label=(An\arabic{*}),leftmargin=3em]
\item
 $f_n,f\in L^2(0,T;H)$ with $f_n\to f$ in $L^2\big(0,T;X^\ast)$,
\item
 $g_n,g\in H$ with $g_n\to g$ in $H$,
\item
$a_n:V\times V\to \rr$ be a bilinear form such that 
for  all $j,k\le N$,
\begin{align}
\big|(a_n-a)(\varphi_j,\varphi_k)\big|
&\to 0.\label{cond-an}
\end{align}
\end{enumerate}
Then the sequence of unique weak solutions $v_n\in W^1(0,T;X,H)$ of
\begin{equation}\label{eq-un}
\dot{v}_n + \OA_n v_n = f_n,\quad v_n(0)=g_n
\end{equation}
converges strongly in $L^2\big(0,T;X)\cap C(0,T;H)$ to the unique weak solution $v \in W^1(0,T;X,H)$ of
\begin{equation}\label{eq-gwu}
\dot{v} + \OA v = f,\quad v(0)=g.
\end{equation}
\end{lemma}
The proof is provided in Section \ref{sec-proof-lem-conv1}.

Profound convergence analysis providing asymptotic convergence rates for the fully discrete scheme can be derived based on the techniques presented by \cite{SchwabWavelet03}, which is beyond the scope of the present article.


\subsection{The symbol method}\label{sec-symbol-method}

The key component of a Galerkin FEM solver is the model dependent stiffness matrix. Using expression~\eqref{eq-ajump} of section~\ref{sec:FlexibleImplementation} above, the entries of that matrix can be derived. The existence of the L{\'e}vy measure $F$ in that expression, however, renders the numerical derivation of the matrix rather cumbersome. Additionally,   the empirical accuracy study of section~\ref{sec:AccuracyStudy} emphasize that utmost care must be taken when the stiffness matrix entries are numerically derived. Consequently, in this section we approach the calculation of that FEM solver components differently. The Fourier approach indicated by Lemma~\ref{lem:SymbolMethod} will allow us to access the model information required for the stiffness matrix and all other FEM solver components not via the operator but rather via the associated symbol, instead. In stark contrast to the operator, the symbol of a L{\'e}vy model is numerically accessible in a unified way for a large set of underlying models and we will present several examples.

Let us state the core lemma of this section.
Here we concentrate on basis functions obeying a simple nodal translation property, which is in particular satisfied for the classical piecewise polynomial basis functions. The translation property can also easily been extended to that one satisfied by biorthogonal wavelet bases that haven prove to be useful for solving Komlogorov PIDEs for option pricing, compare \cite{MatachePetersdorffSchwab2004}.

\begin{lemma}[Symbol method for bilinear forms]
\label{lem:GeneralStiffnessSymbol}
Let the assumptions of Lemma~\ref{lem:SymbolMethod} be satisfied with $\eta=0$. Assume further for $N\in\IN$ a set of functions $\varphi_0,\varphi_1,\dots,\varphi_N\in H_0^{\alpha/2}(\IR)$ and nodes $x_1,\dots,x_N\in\IR$, such that for all $j=1,\dots,N$
	\begin{equation*}
		\varphi_j(x) = \varphi_0(x-x_j),\qquad \forall x\in\IR,
	\end{equation*}
holds. Then we have
	\begin{equation}
	\label{eq:SymbGenClaim1}
		a(\varphi_l,\varphi_k) = \frac{1}{2\pi}\int_\IR A(\xi)e^{i\xi(x_l-x_k)}\left|\widehat{\varphi_0}(\xi)\right|^2\d{\xi}.
	\end{equation}
for all $k,l=1,\dots,N$. If additionally
	\begin{equation}
	\label{eq:GeneralStiffnessSymbolReCond}
		\Re(A(\xi)) = \Re(A(-\xi))\quad\text{ and }\quad\Im(A(\xi)) = -\Im(A(-\xi)),
	\end{equation}
then
	\begin{equation}
	\label{eq:SymbGenClaim2}
		a(\varphi_l,\varphi_k) = \frac{1}{\pi}\int_0^\infty \Re\left(A(\xi)e^{i\xi(x_l-x_k)}\right)\left|\widehat{\varphi_0}(\xi)\right|^2\d{\xi}
	\end{equation}	
for all $k,l=1,\dots,N$.
\end{lemma}
\begin{proof}
Elementary properties of the Fourier transform yield
	\begin{equation}
	\label{eq:ProofSymbGen02}
		\widehat{\varphi_j}(\xi) = e^{i\xi x_j}\widehat{\varphi_0}(\xi),\qquad\forall\xi\in\IR,
	\end{equation}
where
	\begin{equation}
		\widehat{\varphi_0}(\xi) = \frac{2}{\xi^2 h}(1-\cos(\xi h)),\qquad\forall\xi\in\IR.
	\end{equation}
Since $\varphi_j\in H_0^{\alpha/2}(\IR)$, for all $j=1,\dots,N$, the identity~\eqref{eq:SymbGenClaim1}
follows from Lemma~\ref{lem:GeneralStiffnessSymbol} above or Theorem~4.1 and Remark~5.2 and the lines thereafter in \cite{EberleinGlau2012}, respectively. The second claim~\eqref{eq:SymbGenClaim2} is then elementary.
\end{proof}

\begin{corollary}[Symbol method for stiffness matrices]
\label{cor:SymbolMethodForStiffness}
Let $A\in S_\alpha^0$ be a univariate symbol with associated operator $\mathcal{A}$. Denote by $\varphi_j\in L^1(\IR)$, $j\in 1,\dots,N$ the basis functions of a Galerkin scheme associated with an equidistantly spaced grid $\Omega = \{x_1,\dots,x_N\}$ possessing the property
	\begin{equation}
	\label{eq:phiiphi0}
		\varphi_j(x) = \varphi_0(x-x_j),\qquad\forall x\in\IR,
	\end{equation}
for some $\varphi_0:\IR\rightarrow\IR$. Then, the stiffness matrix $A\in\IR^{N\times N}$ of the scheme can be computed by
	\begin{equation}
	\label{eq:StiffnessSymbol}
		A_{kl} = \frac{1}{2\pi}\int_\IR A(\xi)e^{i\xi(x_l-x_k)}\left|\widehat{\varphi_0}(\xi)\right|^2\d{\xi}
	\end{equation}	
for all $k,l=1,\dots,N$.
\end{corollary}

\begin{proof}
The proof is an immediate consequence of Lemma~\ref{lem:GeneralStiffnessSymbol}.
\end{proof}

\begin{remark}[On the symbol method for bilinear forms]
\label{rem:SymbolForBilinearForms}
From a numerical perspective, the representations of the stiffness matrix entries provided in Lemma~\ref{lem:GeneralStiffnessSymbol} and Co\-rol\-lary~\ref{cor:SymbolMethodForStiffness} are highly promising: 
\begin{enumerate}
\item Instead of the double integrals appearing in \eqref{eq-ajump}, only one dimensional integral need to be computed.
\item The model specific information is expressed via the symbol $\xi\mapsto A(\xi)$, which for a large set of models is available in form of an explicit function of $\xi$  and the model parameters, a feature that we now can exploit numerically. We give a short list of examples below. For further examples we refer to \cite{Glau2013} and \cite{Glau2016}.
\item
Representation \eqref{eq:StiffnessSymbol} displays the 
entries of the stiffness matrix as \emph{Fourier integrals}. Moreover, the nodes appear as Fourier variables. As a consequence, Fast Fourier Transform (FFT) methods can be used to accelerate their simultaneous computation.
\end{enumerate}

\end{remark}

Expression~\eqref{def-opA} introduces operators $\mathcal{A}$ for L{\'e}vy processes $L$ in terms of the characteristic triplet $(b,\sigma,F)$. The following examples present the respective symbols for some well known L{\'e}vy models.

\begin{example}[Symbol in the Black-Scholes (BS) model]
In the univariate Black-Scholes model, determined by the Brownian volatility $\sigma>$, the symbol is given by
\label{ex:BSsymbol}
	\begin{equation}
	\label{eq:defbssymbol}
		A(\xi) = A^\text{bs}(\xi) = i\xi b + \frac{1}{2}\sigma^2\xi^2,
	\end{equation}
with drift set $b$ to
	\begin{equation}
		b = r-\frac{1}{2}\sigma^2.
	\end{equation}
\end{example}

\begin{example}[Symbol in the Merton model]
\label{ex:defMertonsymbol}
In the Merton model where $\sigma>$, $\lambda>0$, $\alpha\in\IR$ and $\beta>0$, the symbol computes to
	\begin{equation}
		A(\xi) = A^\text{merton}(\xi) = \frac{1}{2}\sigma^2\xi^2 + i\xi b - \lambda\left(e^{-i\alpha\xi-\frac{1}{2}\beta^2\xi^2}-1\right)
	\end{equation}
	with drift set to
        \begin{equation}
               b = r-\frac{1}{2}\sigma^2-\lambda\left(e^{\alpha+\frac{\beta^2}{2}}-1\right),
        \end{equation}
as required by the no-arbitrage condition.
\end{example}

\begin{example}[Symbol in the CGMY model]
\label{ex:CGMYsymbol}
In the CGMY model of \cite{CarrGemanMadanYor2002} with $\sigma>0$, $C>0$, $G\geq 0$, $M\geq 0$ and $Y\in(1,2)$, the symbol computes to
	\begin{multline}
	\label{eq:defcgmysymbol}
		A(\xi) = A^\text{cgmy}(\xi) = i\xi b + \frac{1}{2}\sigma^2\xi^2 \\
		-C\Gamma(-Y) \left[(M+i\xi)^Y - M^Y + (G-i\xi)^Y - G^Y\right],
	\end{multline}
for all $\xi\in\IR$, with drift $b$ set to
	\begin{equation}
		b = r-\frac{1}{2}\sigma^2-C\Gamma(-Y)\left[(M-1)^Y - M^Y + (G+1)^Y-G^Y\right]
	\end{equation}
for martingale pricing.
\end{example}

\begin{example}[Symbol in the NIG model]
\label{ex:NIGsymbol}
With $\sigma>0$, $\alpha>0$, $\beta\in\IR$ and $\delta>0$ such that $\alpha^2>\beta^2$, the symbol of the NIG model is given by
	\begin{equation}
	\label{eq:defnigsymbol}
		A(\xi) = A^\text{nig}(\xi) = \frac{1}{2}\sigma^2\xi^2 + i\xi b -\delta\left(\sqrt{\alpha^2 - \beta^2} - \sqrt{\alpha^2 - (\beta-i\xi)^2}\right)
	\end{equation}
for all $\xi\in\IR$ with drift given by
	\begin{equation}
		b = r-\frac{1}{2}\sigma^2 - \delta\left(\sqrt{\alpha^2-\beta^2}-\sqrt{\alpha^2-(\beta+1)^2}\right)
	\end{equation}
\end{example}

Implementing \eqref{eq:StiffnessSymbol}, we encounter new numerical challenges: From the perturbation study in Section \ref{sec:AccuracyStudy} that we need to evaluate the integrals at high precision.
Consider first the Black-Scholes model and choose the piecewise linear hat functions as basis elements as a toy example. Applying a standard Matlab integration routine will lead to considerable errors. To understand the effect, let us have a closer look at the integrands in \eqref{eq:StiffnessSymbol}, which we depict in Figure\tild\ref{fig:phi0hatFourier}.

 
\begin{figure}[h!]
\begin{center}
\makebox[0pt]{\includegraphics[scale=1]{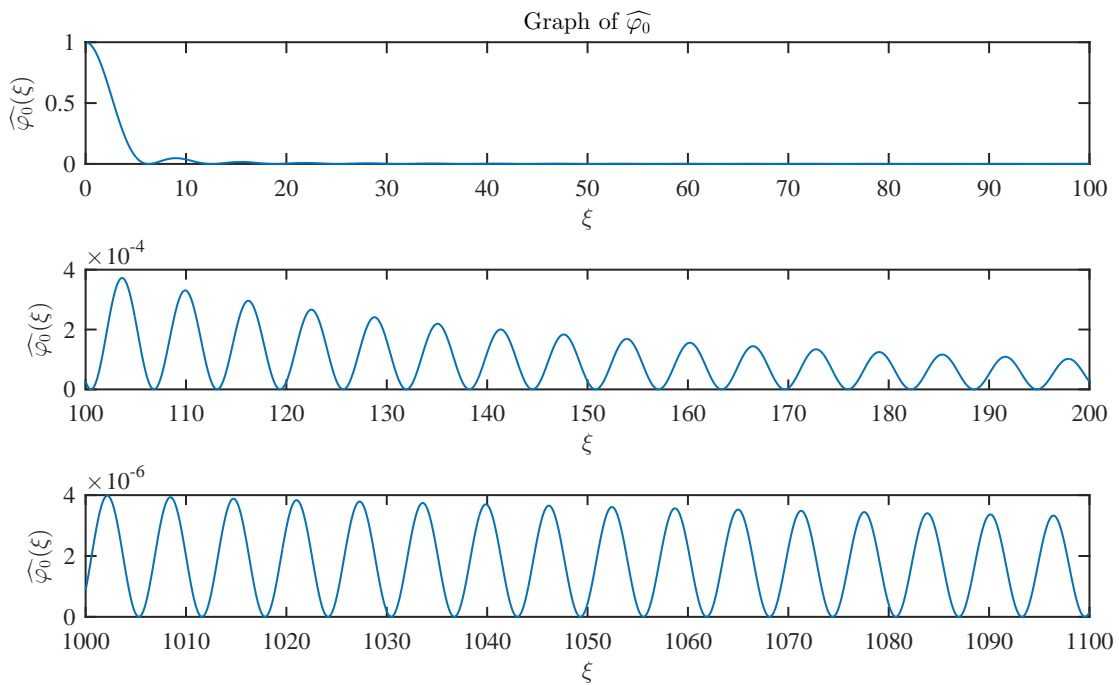}}
\caption{Graph of $\widehat{\varphi_0}$, the Fourier transform of the hat function $\varphi_0$ centered over the origin, evaluated over three subintervals of $\IR^+$. The mesh is chosen with $h=1$. The oscillations and the rather slow decay to zero complicate numerical integration with high accuracy requirements considerably when $\widehat{\varphi_0}$ is involved.}
\label{fig:phi0hatFourier}
\end{center}
\end{figure}
\begin{figure}
\begin{center}
\makebox[0pt]{\includegraphics[scale=1]{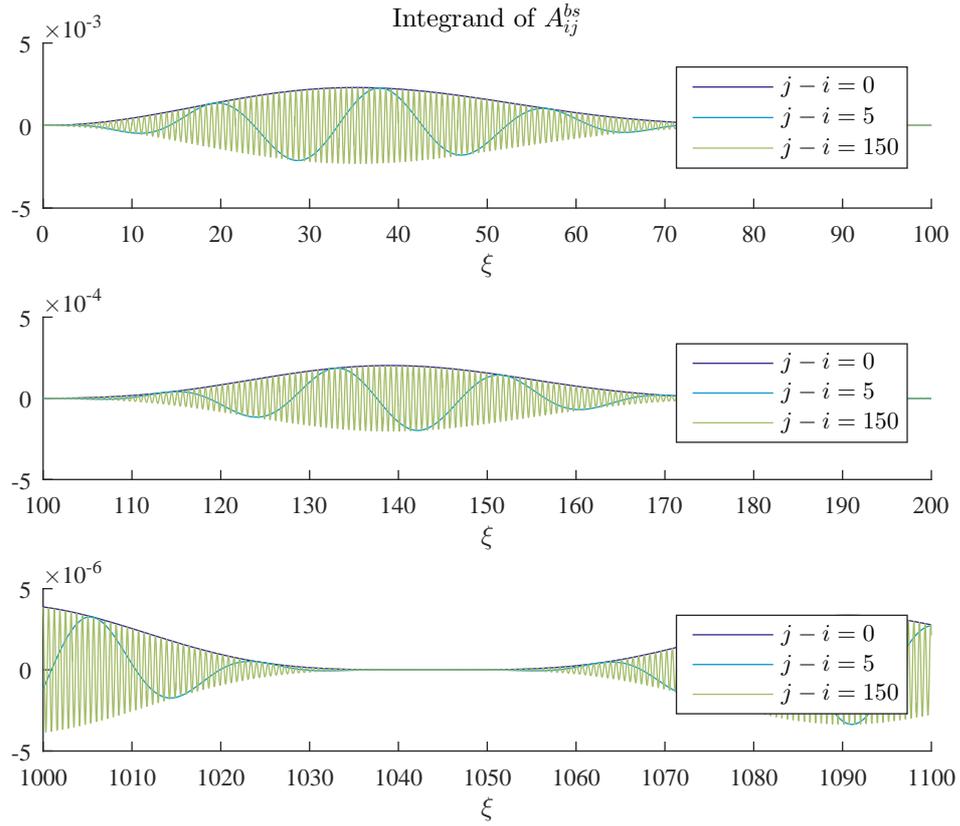}}
\caption{The integrand for the Black-Scholes stiffness matrix $A_{kl}$ for several values of $l-k$. The grid of the hat functions spans the interval $[-5, 5]$ with $150$ equidistantly spaced inner nodes and grid fineness $h=0.0662$. A Black-Scholes solution on this grid would thus be represented by the weighted sum of $150$ hat functions. We observe that oscillations of the integrand increase in the value of $|l-k|$.}
\label{fig:HatBSIntegrands}
\end{center}
\end{figure}
More precisely, we show several integrands of $A\in\IR^{N\times N}$ in the representation provided by~\eqref{eq:StiffnessSymbol} of Corollary~\ref{cor:SymbolMethodForStiffness} with the symbol given by Example~\ref{ex:BSsymbol}. Each integrand is evaluated for a different value of $l-k$ over three different subintervals taken from the unbounded integration range. In the Fourier approach of calculating the stiffness matrix $A\in\IR^{N\times N}$ via the respective symbol, the integrands of $A_{kl}$ would have to be numerically integrated for all $l-k\in\{-(N-1),\dots,-1,0,1,\dots,N-1\}$.

The larger $|j-i|$, however, the more severe the numerical challenges for evaluating the integrand, as Figure~\ref{fig:HatBSIntegrands} demonstrates. All integrands illustrated therein decay rather slowly. Additionally, oscillations increase in $|j-i|$. In combination, these two observation seriously threaten a numerically reliable evaluation of the integral. With increasing values of $|j-i|$, the oscillations of the integrand accelerate and the number of necessary supporting points for accurate integration soars.


These considerations show us that we need to further investigate the problem to obtain a flexible method to compute the stiffness matrix reliably and with low computational cost. The path that we propose here is to modify the problem in such a way that the resulting integrands decay much faster so that the domain of integration can be chosen considerably smaller and a usual integration routine such as Matlab's function \texttt{quadgk} is sufficient. To do so, we first observe that the hat functions, which we used in our toy example, are piecewise linear functions. While being continuous they are not continuously differentiable everywhere and thus lack smoothness on an elementary level already. This lack of smoothness translates into a slow decay of their Fourier transform. 

Therefore, we propose to replace the piecewise linear basis functions by basis functions that display considerably higher regularity leading to appealing decay properties of the integrands in \eqref{eq:StiffnessSymbol}. In the following two sections, we present two different approaches to implement such a problem modification.



\subsubsection{From classic basis functions to mollified hats}
It is well known that convolution with a smooth function has a smoothing effect on the function that the convolution is applied to. Our first basis function alternative will therefore be a classic hat function smoothed by convolution. Functions that qualify for this smoothing by convolution are called mollifiers. 
 In order to choose an appropriate mollifier for our purposes---the fast and accurate computation of the integrals in \eqref{eq:StiffnessSymbol}, the mollifiers need to display two essential features:
 \begin{itemize}
 \item[(1)] The Fourier transform of the modified basis function needs to be available. 
 \item[(2)] The smoothing effect needs to be steerable through a parameter.
 \end{itemize}
 As the Fourier transform of the convolution of two functions is the product of the two Fourier transformed functions, (2) boils down to the availability of the Fourier transform of the mollifier. Since the Fourier transform of standard mollifiers is not available in closed form, we widen the range of the standard mollifiers and allow for non-compact support.
More precisely,  we call the sequence $m=(m_k)_{k\in\IN}$, $m_k\in L^1(\IR^d)$ for all $k\in\IN$, a \textit{mollifier}, if
	\begin{enumerate}
		\item $m_k \geq 0$, for all $k\in\IN$,
		\item $\int_{\IR^d} m_k(x)\d{x} = 1$, and
		\item for all $\varrho>0$ we have the convergence
				$\int_{\IR^d\backslash B_\varrho(0)}m_k(x)\d{x} \rightarrow 0$ for $k\rightarrow \infty$.
			
	\end{enumerate}
	
Feature (1) is often required and we follow the usual construction here. By Proposition and Definition~2.14 in \cite{AltFunktionalanalysis2011} we can adjust the influence of mollification by a mollification $\varepsilon$.
To this end let $m\in L^1(\IR^d)$ with
	\begin{equation}
		m \geq 0,\quad\text{ and }\quad\int_{\IR^d}m(x)\d{x}=1.
	\end{equation}
Define
	\begin{equation}
	\label{eq:MolliEpsilon}
		m^\varepsilon = \frac{1}{\varepsilon^d}m\left(\frac{\cdot}{\varepsilon}\right).
	\end{equation}
Then for each $\varrho>0$ we have
	$
		\int_{\IR^d}m^\varepsilon(x)\d{x} = 1\text{ and } \int_{\IR^d\backslash B_\varrho(0)} m^\varepsilon(x)\d{x} \rightarrow 0$
for $\varepsilon\rightarrow 0$. Consequently, for each null sequence $(\varepsilon_k)_{k\in\IN}$ the sequence $(m^{\varepsilon_k})_{k\in\IN}$ is a mollifier in the sense of our definition.

\begin{example}[A mollifier based on the Normal distribution]
\label{ex:GaussianDiracSequence}
We present an example for a mollifier. Define 
	\begin{equation}
	\label{eq:defExDiracGauss}
		m_\text{Gaussian} (x)= \frac{1}{\sqrt{2\pi}} e^{-\frac{x^2}{2}}.
	\end{equation}
Then we call $(m^{\varepsilon_k}_\text{Gaussian})_{k\in\IN}$ defined according to~\eqref{eq:MolliEpsilon} 
a \emph{Gaussian mollifier}.
The characteristic function of the Gaussian mollifier is known in closed form,	\begin{equation}
	\label{eq:defExQuasiMolliGaussFourier}
		\widehat{m_\text{Gaussian}^\varepsilon}(\xi) = \exp\left(-\frac{1}{2} \varepsilon^2\xi^2\right),
	\end{equation}
thus exhibiting exponential decay.
\end{example}


It is a well known result, that mollified functions $f\conv m_k$ converge to $f$ in $L^p(\IR^d)$, $1\leq p<\infty$ when $k$ tends to infinity, see for example Satz~2.15 in \cite{AltFunktionalanalysis2011}.



\begin{figure}
\begin{center}
\makebox[0pt]{\includegraphics[scale=1]{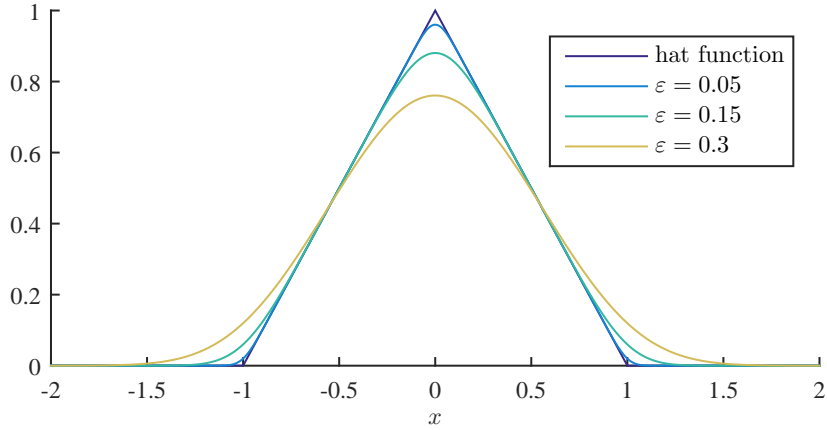}}
\caption{A comparison between the classic hat function $\varphi_0$ on a grid with $h=1$ and the mollified hat function $\varphi^\varepsilon_0 = \varphi_0 \conv m_\text{Gaussian}^\varepsilon$ for several values of $\varepsilon\in \{0.05, 0.15, 0.3\}$ using the Gaussian mollifier of Example~\ref{ex:GaussianDiracSequence}.}
\label{fig:exfemhatnormalmollified}
\end{center}
\end{figure}
\begin{figure}[t]
\begin{center}
\makebox[0pt]{\includegraphics[scale=1]{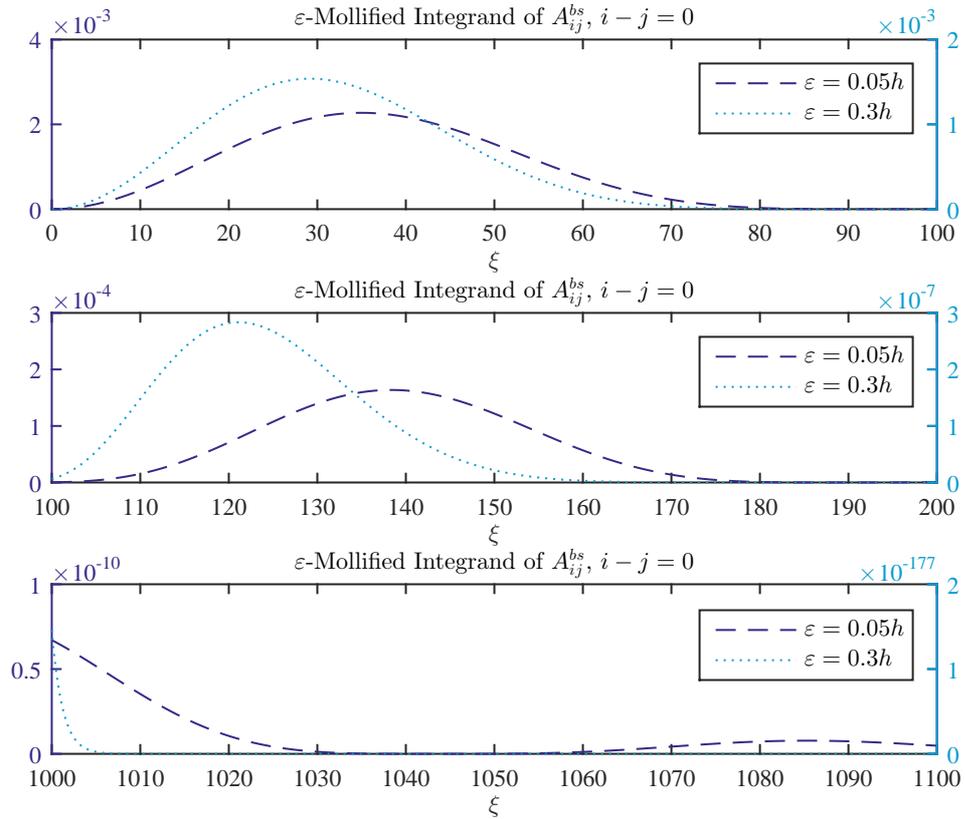}}
\caption{The integrand of $A_{kl}$, the stiffness matrix of the Black-Scholes model with mollified hat functions as basis functions for the main diagonal entry, $l-k=0$.}
\label{fig:HatBSIntegrandsMollified}
\end{center}
\end{figure}



Figure~\ref{fig:HatBSIntegrandsMollified} displays the integrand in the stiffness matrix of the Black-Scholes model. The integrand is evaluated on three subintervals of the semi-infinite integration region. The grid setting is identical to the one of Figure~\ref{fig:HatBSIntegrands}. Instead of classic hat functions their mollified counterparts have been employed as basis functions using the Gaussian mollifier of Example~\ref{ex:GaussianDiracSequence} as smoothing influence. Even with just a slight mollification influence, $\varepsilon = 0.05h$, the the integrand decays noticeably faster thanks to the mollification. For moderate values of $\varepsilon=0.3h$ the integrand decays to zero rapidly.

Before we test the performance of this approach to modify the Galerkin scheme in Section \ref{sec-num} below, we introduce our second approach.

\subsubsection{Splines as basis functions}
Instead of mollification of piecewise linear basis functions, we can choose basis functions that display higher regularity itself.
We therefore investigate a well-established class of Finite Element basis functions as candidates for our purposes, namely cubic splines. 
 Spline theory is a well-investigated field that applies to a much broader context than we consider here. We refer the reader to \cite{Schumaker07} for an introduction and overview. From our perspective, splines are smooth basis functions. Their Fourier transform is accessible and the theory of function spaces they span is well-established. We give the definition of the Irwin-Hall cubic spline that inherits is name from related probability distribution.

We define the univariate \emph{Irwin-Hall spline} $\varphi:\IR\rightarrow\IR^+$ by
	\begin{equation}
	\label{eq:defSpline}
	\varphi(x) = \frac{1}{4}\begin{cases}
		(x+2)^3 & ,\ -2\leq x < -1\\
		3|x|^3 - 6x^2 + 4 & ,\ -1\leq x < 1\\
		(2-x)^3 & ,\ 1\leq x \leq 2\\
		0 & ,\ \text{ elsewhere}
		\end{cases}
	\end{equation}
for all $x\in\IR$. The spline $\varphi$ has compact support on $[-2,2]$ and is a cubic spline. We use it to define a spline basis:

\begin{definition}[Spline basis functions on an equidistant grid]
\label{def:FEMsplinefunctions}
Choose $N\in\IN$. Assume an equidistant grid $\Omega=\{x_1,\dots,x_N\}$, $x_j\in\IR$ for all $j=1,\dots,N$, with mesh fineness $h>0$. Let $\varphi$ the Irwin-Hall spline of \eqref{eq:defSpline}. For $j=1,\dots,N$ define
	\begin{equation*}
		\varphi_j(x) = \varphi((x-x_j)/h),\qquad \forall x\in\IR.
	\end{equation*}
We call $\varphi_j$ \emph{the spline basis function associated to node $j$}.
\end{definition}

For a given grid $\Omega=\{x_1,\dots,x_N\}$, $x_j\in\IR$, Definition~\ref{def:FEMsplinefunctions} provides the set of spline basis functions that we also use in our numerical implementation, later. In standard spline literature, the set of Irwin-Hall spline basis functions is usually enriched with additional splines associated with the first and the last node of the grid that provide further flexibility in terms of Dirichlet and Neumann boundary conditions, see for example~\cite{Schumaker07}. We omit the three Irwin-Hall basis functions associated with either of the first and the last grid nodes thus implicitly prescribing Dirichlet, Neumann and second order derivative zero boundary conditions. 

\begin{lemma}[Fourier Transform of the Irwin-Hall spline]
\label{lem:IHsplineFourier}
Let $\varphi$ be the Irwin-Hall cubic spline of \eqref{eq:defSpline}. Then its Fourier transform $\widehat{\varphi}$ is given by
	\begin{equation}
		\widehat{\varphi}(\xi) = \frac{3}{\xi^4}\left(\cos(2\xi) - 4\cos(\xi) +3\right)
	\end{equation}
for all $\xi\in\IR$.
\end{lemma}
The proof of the Lemma follows by elementary calculation. This immediately gives the following

\begin{corollary}[Fourier Transform of spline basis functions on an equidistant grid]
\label{cor:FTsplines}
Choose $N\in\IN$. Assume an equidistant grid $\Omega=\{x_1,\dots,x_N\}$, $x_j\in\IR$ for all $j=1,\dots,N$, with mesh fineness $h>0$ and let $\varphi_j$ be the spline basis function associated with node $j$ as defined in Definition~\ref{def:FEMsplinefunctions}. Its Fourier transform is given by
	\begin{equation*}
		\widehat{\varphi_j}(\xi) = e^{i\xi x_j} \frac{3}{h^3 \xi^4}(\cos(2\xi h) - 4\cos(\xi h) +3)
	\end{equation*}
for all $\xi\in\IR$.
\end{corollary}



\begin{figure}[t]
\begin{center}
\makebox[0pt]{\includegraphics[scale=1]{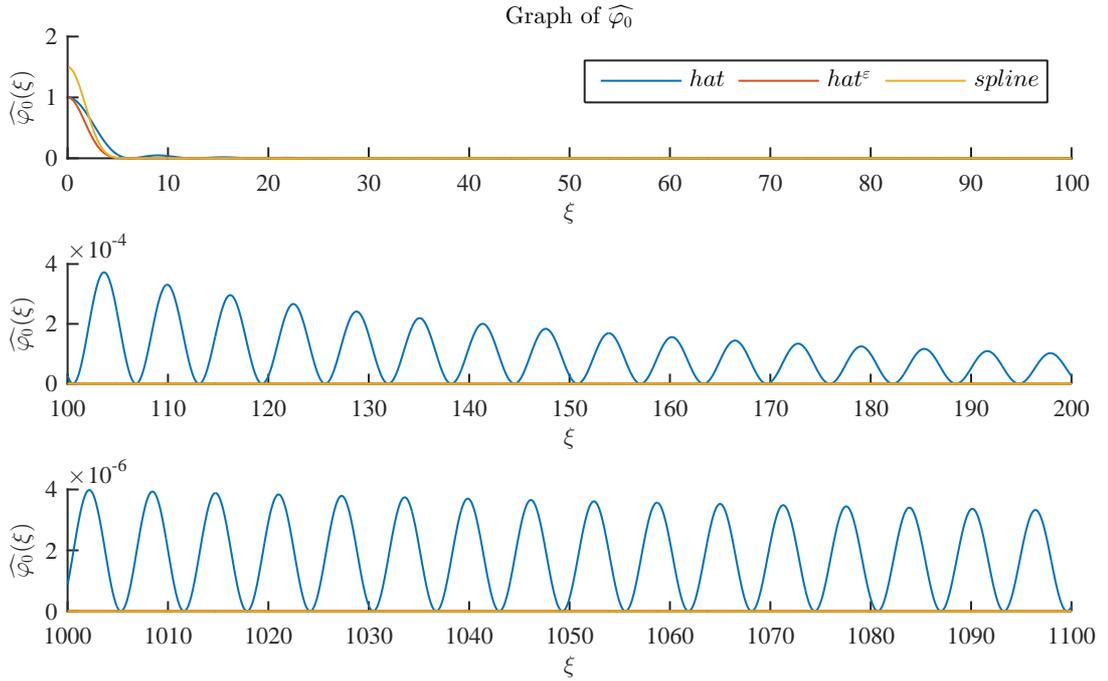}}
\caption{Graphs of the Fourier transforms of all basis function candidates presented in this section, evaluated over three subintervals of $\IR^+$. The mesh is chosen with $h=1$, the mollification parameter is again set to $\varepsilon=0.3h$.}
\label{fig:phi0allFourier}
\end{center}
\end{figure}

Figure~\ref{fig:phi0allFourier} compares the decay behaviour of the Fourier transforms of all three basis function types. As Figure~\ref{fig:phi0hatFourier} already illustrated, the Fourier transform of the classic hat functions exhibits both slow decay rates and oscillatory behaviour. In stark contrast the Fourier transforms of the mollified hats as well as the Fourier transform of Irwin-Hall splines visually decays to zero instantly. In case of the mollified hat functions this is due to the exponential decay of the Fourier transform of the Gaussian mollifier while for splines Corollary~\ref{cor:FTsplines} displays a polynomial decay of order $4$. In this regard, both alternatives to the classic hat functions are promising candidates for the implementation of the symbol method of Corollary~\ref{cor:SymbolMethodForStiffness}. In the upcoming Section\tild\ref{sec-num} we put that promise to the test.

\FloatBarrier
\section{Numerical Implementation}\label{sec-num}
In this section we implement the pricing PIDEs for plain vanilla call and put options and test the two approaches to the symbol method experimentally.

\begin{theorem}[Feynman-Kac]
\label{thrm:FeynmanKac}
Let $(L_t)_{t\geq 0}$ be a (time-homogeneous) L{\'e}vy process. 
Consider the PIDE
	\begin{equation}
	\label{eq:PIDE}
	\begin{split}
		\partial_t V^{C,P} + \mathcal{A} V^{C,P} =&\ 0,\qquad\text{for almost all $t\in(0,T)$}\\
		V^{C,P}(0) =&\ g^{C,P},
	\end{split}
	\end{equation}
where $\mathcal{A}$ is the operator associated with the symbol of $(L_t)_{t\geq 0}$. Assume further the assumptions (A1)--(A3) of \cite{EberleinGlau2012} to hold. Then~\eqref{eq:PIDE} possesses a unique weak solution
	\begin{equation}
		V^{C,P} \in W^1(0,T;H^{\alpha/2}_\eta(\IR^d), L^2_\eta(\IR^d))
	\end{equation}
where $\alpha>0$ is the Sobolev index of the symbol of $(L_t)_{t\geq 0}$ and $\eta\in\IR^d$ is chosen according to Theorem~6.1 in \cite{EberleinGlau2012}. If additionally $g^{C,P}_\eta\in L^1(\IR^d)$, then the relation
	\begin{equation}
		V^{C,P}(T-t,x)=\IE\left[g^{C,P}(L_{T-t}+x)\right]
	\end{equation}
holds for all $t\in [0,T]$, $x\in\IR^d$.
\end{theorem}

\begin{proof}
The result proved in \cite{EberleinGlau2012} and follows from Theorem~6.1 therein.
\end{proof}

\begin{remark}
Setting $g^{C,P} = g^C$ in \eqref{eq:PIDE}, the payoff profile of a European call option, results in $V^C$ being the price of a European call option. Analogously, setting $g^{C,P} = g^P$, the payoff profile of a European put option, results in $V^P$ being the price of a European call option.
\end{remark}

Algorithm~\ref{algo:SymbolMethod} summarizes the abstract structure of a general FEM solver based on the symbol method. By plugging the symbol associated to the model of choice into the computation of line~\ref{line:SymbolStiffness} of the algorithm, the solver instantly adapts to that model. In other words, only one line needs to be specified to obtain a model specific solver for option pricing. As Examples~\ref{ex:BSsymbol}, \ref{ex:defMertonsymbol}, \ref{ex:CGMYsymbol}, \ref{ex:NIGsymbol} and others emphasize, the symbol exists in analytically (semi--)closed form for many models, indeed. Algorithm~\ref{algo:SymbolMethod} thus provides a very appealing tool for FEM pricing in practice. 

\begin{algorithm}[ht!]
\caption{A symbol method based FEM solver}\label{algo:SymbolMethod}
\begin{algorithmic}[1]
\State Choose equidistant space grid $x_i$, $i=1,\dots,N$
\State Choose basis functions $\varphi_i$, $i=1,\dots,N$, with $\varphi_i(x) = \varphi_0(x-x_i)$ for some $\varphi_0$
\State Choose equidistant time grid $T_j$, $j=0,\dots,M$
\Procedure{Compute Mass Matrix $M$}{}
\State Derive the mass matrix $M\in\IR^{N\times N}$ by 
\State $\quad M_{kl} = \int_\IR \varphi_l(x)\varphi_k(x)\d{x},\qquad \forall k,l=1,\dots,N$
\EndProcedure
\Procedure{Compute Stiffness Matrix $A$}{}
\State \parbox[t]{\dimexpr\linewidth-\algorithmicindent}{Derive the stiffness matrix $A\in\IR^{N\times N}$ by plugging the symbol $A$ of the chosen model into the following formula and computing\strut}
\State $\quad A_{kl} = \frac{1}{2\pi}\int_\IR A(\xi)\, e^{i\xi(x_k-x_l)}\,|\widehat{\varphi_0}(\xi)|^2\,\d{\xi},\qquad \forall k,l=1,\dots,N$\label{line:SymbolStiffness}
\State using numerical integration
\EndProcedure
\Procedure{Run Theta Scheme}{}
\State \parbox[t]{\dimexpr\linewidth-\algorithmicindent}{
Choose a function $\psi$ to subtract from the original pricing problem to obtain a zero boundary problem and retrieve the respective basis function coefficient vectors $\overline{\psi}^k\in\IR^N$, $k=1,\dots,M$.
Consider the suggestions by Lemma~\ref{lem:SubtractingBSprices} or Lemma~\ref{lem:SubtractingQHSnormal} for plain vanilla European options below.\strut}
\State \parbox[t]{\dimexpr\linewidth-\algorithmicindent}{Choose an appropriate basis function coefficient vector $V^1\in\IR^N$ matching the initial condition of the transformed problem\strut}
\State \parbox[t]{\dimexpr\linewidth-\algorithmicindent}{Derive the right hand side vectors $F^k\in\IR^N$, $k=0,\dots,M$. Consult Lemma~\ref{lem:SubtractingBSprices} or Lemma~\ref{lem:SubtractingQHSnormal} matching the choice of $\psi$.\strut}
\State \parbox[t]{\dimexpr\linewidth-\algorithmicindent}{Choose $\theta\in[0,1]$ and run the iterative scheme\strut}
\State \textbf{for} $k=0:(M-1)$
\State $\quad\quad V^{k+1} = (M + \Dt\,\theta\,A)^{-1}\left(\left(M-\Dt\,(1-\theta)\,A\right)V^k + F^{k+\theta}\right)$
\State \textbf{end}
\EndProcedure
\Procedure{Reconstruct Solution to Original Problem}{}
\State \parbox[t]{\dimexpr\linewidth-\algorithmicindent}{Add previously subtracted right hand side $\psi$ to the solution of the transformed problem by computing\strut}
\State $\quad\widetilde{V}^k = V^k + \overline{\psi}^k,\qquad k=0,\dots,M$
\State \parbox[t]{\dimexpr\linewidth-\algorithmicindent}{to retrieve the basis function coefficient vectors $\widetilde{V}^k$, $k=0,\dots,M$, to the original pricing problem\strut}
\EndProcedure
\end{algorithmic}
\end{algorithm}



\subsection{Truncation to zero boundary conditions}
\label{susec:trunc}
As we derive prices of plain vanilla European call and put options, the solution to the respective pricing PIDE is defined on the whole real line. As a first step towards a discretization, we want to truncate the domain to bounded interval $(a,b)$ and we choose to implement zero boundary conditions. To do so, we follow the standard procedure to subtract an appropriate auxiliary function $\psi$. Having chosen $\psi$, the resulting modified problem for $\phi= V^{C,P}-\psi$ is


	\begin{equation}
	\label{eq:glrknpricingPIDEmodi}
	\begin{alignedat}{2}
		\partial_t \phi(t,x) + \left(\mathcal{A}\phi\right)(t,x) =&\ f(t,x),\qquad &&\forall(t,x)\in (0,T)\times \IR\\
		\phi(0,x) =&\ g_\Psi(x),\qquad &&\forall x\in\IR,
	\end{alignedat}
	\end{equation}
where $g_\Psi(x) = g(x)-\psi(0,x)$ for all $x\in\IR$ and the right hand side $f$ is given by
	\begin{equation*}
		f(t,x) := -\left(\partial_t\psi(t,x) + (\mathcal{A}\psi)(t,x)\right).
	\end{equation*}
The solution $V^{C,P}$ to the original problem~\eqref{eq:PIDE} can easily be restored by $V^{C,P}=\phi + \psi$. We establish the properties that $\psi$ needs to provide, later, where we will present some examples, as well.

The right hand side in vector notation is given by $F(t^k) =(F_1(t^k),\dots,F_N(t^k))\in \IR^N$ for each $t^k$ on the time grid with $F_j(\cdot)$, $j=1,\dots,N$, given by
	\begin{equation}
	\label{eq:rhsMerton01}
	\begin{split}
		F_j =&\ -\int_\mathbb{R}\left(\partial_t\psi(t,x)+(\mathcal{A}\psi)(t,x)+r\psi(t,x)\right)\varphi_j(x)\d{x}
	\end{split}
	\end{equation}
for all $j=1,\dots,N$.

We observe that the operator $\OA$ applied to the auxiliary function $\psi$ appears in the integral. For the same reasons as in the computation of the stiffness matrix entries, we decide to apply the symbol method for the computation of the entries of the right hand side $F\in\IR^N$. We pursue these considerations in following section.

\subsection{Computation of the right hand side $F$}
\label{subsec:RhsFourier}
First, we need to choose an appropriate auxiliary function $\psi$. As its purpose is to enable us to truncate the domain and insert zero boundary conditions, we need to inspect the limit behaviour of  
the price value 
	\begin{equation}	
	\label{eq:Ecallasymptotic}	
	\begin{alignedat}{2}	
		&V^C(x, t) \rightarrow 0,\qquad &x\rightarrow -\infty,\ t\in[0,T]\\
		&V^C(x, t) \rightarrow e^x - Ke^{-rt},\qquad &x\rightarrow +\infty,\ t\in[0,T]
	\end{alignedat}
	\end{equation}	
for call options and
	\begin{equation}	
	\label{eq:Eputasymptotic}
	\begin{alignedat}{2}	
		&V^P(x, t) \rightarrow Ke^{-rt}-e^x,\qquad &x\rightarrow -\infty,\ t\in[0,T]\\
		&V^P(x, t) \rightarrow 0,\qquad &x\rightarrow +\infty,\ t\in[0,T]
	\end{alignedat}
	\end{equation}
for put options.  This is the usual way to obtain the auxiliary function. Now, in regard to our specific approach, relying on the Fourier transforms, we identify additional desirable feature for the auxiliary function. We denote  $\widehat{\psi}(t,z):= \widehat{\psi(t,\cdot)}(z)$.
Consider a smooth function $\psi:[0,T]\times\IR\rightarrow\IR$ such that $\psi(t)\in H_\eta^{\alpha/2}(\IR)$ for all $t\in[0,T]$ for some $\eta\in\IR$. Then, for the second summand in~\eqref{eq:rhsMerton01} we have by applying the symbol method of Lemma~\ref{lem:SymbolMethod} that 
	\begin{equation}
	\label{eq:rhsMerton05}
	\begin{split}		
		\int_\IR(\mathcal{A}\psi)(t,x)\varphi_j(x)\d{x} =&\ \frac{1}{2\pi}\int_\IR A(\xi-i\eta)\widehat{\psi}(t,\xi-i\eta)\overline{\widehat{{\varphi_j}}(\xi+\eta)}\d{\xi},
	\end{split}
	\end{equation}
where $A$ denotes the symbol of the model. With the above identity, we are able to derive the right hand side $(F_j)_{j=1,\dots,N}$ of~\eqref{eq:rhsMerton01} in terms of Fourier transforms by
	\begin{equation*}
		F_j = -\frac{1}{2\pi}\int_\IR \left(\widehat{\partial_t\psi}(t,\xi-i\eta) +A(\xi-i\eta)\widehat{\psi}(t,\xi-i\eta)+r\widehat{\psi}(t,\xi-i\eta)\right)\overline{\widehat{{\varphi_j}}(\xi + \eta)}\d{\xi}.
	\end{equation*}
This shows that $\psi$ is numerically suitable for the purpose of localizing the pricing PIDE if
	 $\psi$ is quickly evaluable on the region $[a,b]\times[0,T]$ and the integrals determining $F_j$ 
			can be numerically evaluated fast for all $j=1,\dots,N$. The first feature allows retransforming the solution of the localized problem into the solution of the original pricing PIDE, while the second grants the fast numerical derivation of the right hand side in equation \eqref{eq:glrknpricingPIDEmodi}. These considerations lead us to the following list of desirable features for the auxiliary function $\psi$ that is required to obey the respective limit conditions \eqref{eq:Ecallasymptotic}, \eqref{eq:Eputasymptotic}:
\begin{enumerate}
\item a (semi-)closed expression of the function $\psi$, 
\item a (semi-)closed expression of its Fourier transform $\widehat{\psi}$
\item and fast decay of $|\widehat{\psi}(\xi)|$ and $|\widehat{\partial_t\psi}(\xi)|$ for $|\xi|\rightarrow\infty$.
\end{enumerate}
 The smoother $\psi$, the faster $|\widehat{\psi}|$ decays. 
We thus need different functions $\psi$ to subtract that not only fulfil the appropriate boundary conditions~\eqref{eq:Ecallasymptotic} or~\eqref{eq:Eputasymptotic} but that are also as smooth as possible.

In the following two subsections we will analyze two candidates for $\psi$ that display the desired features.



A first suggestion for $\psi$ consists in using Black-Scholes prices as functions in $x=\log(S_0)\in[a,b]$ and time to maturity $t\in[0,T]$ for localization of the pricing PIDE. We express the price of a European option with payoff profile $g^{C,P}$ in the Black-Scholes model in terms of (generalized) Fourier transforms and define $\psi$ accordingly, as the following Lemma explains.

\begin{lemma}[Subtracting Black-Scholes prices]
\label{lem:SubtractingBSprices}
Choose a Black-Scholes volatility $\sigma^2>0$. Define $\psi$ to be the associated Black-Scholes price,
	\begin{equation}
	\label{eq:bsrhsPsi01}
		\psi(t,x) = \psi^{\text{bs}}(t,x)=e^{-\eta x}e^{-rt}\frac{1}{2\pi}\int_\IR e^{i\xi x}{\widehat{g^{C,P}}(-(\xi+i\eta))}\varphi^{\text{bs}}_{t,\sigma}(\xi+i\eta)\d{\xi},
	\end{equation}
with $\varphi^{\text{bs}}_{t,\sigma}(z)  = \ee{t A^{bs}(z)}$. We denote $A$ the symbol of the associated operator $\mathcal{A}$, $r\geq0$ the prevailing risk-free interest rate and choose $\eta<-1$ and $\eta>0$ in the case of a call option and  $\eta$ for the put. Then the right hand side $F:[0,T]\rightarrow\IR^N$ computes to
	\begin{multline}
	\label{eq:FjBlackScholes}
		F_j(t) = \frac{1}{2\pi}\int_\IR\bigg(\left(A^{\text{bs}}-A\right)(\xi-i\eta)\bigg)\\ \widehat{g^{C,P}}(\xi-i\eta)\exp\left(-t\left(r+ A^{\text{bs}}(\xi-i\eta)\right)\right)\overline{\widehat{\varphi_j}(\xi+i\eta)}\d{\xi}
	\end{multline}
for all $j=1,\dots,N$.
\end{lemma}

\begin{proof}
In order to derive the right hand side, we need to represent $\psi$ in Fourier terms. Since for call and put options,
 $\psi\notin L^1(\IR)$, we compute the (generalized) Fourier transform of $\psi$ or the Fourier transform of $\psi_\eta = e ^{\eta\cdot}g$, respectively. We get
	\begin{equation}
	\label{eq:bsrhsPsi02}
	\begin{split}
		\psi_\eta(t, x) 
			=&\ e^{-rt}\frac{1}{2\pi}\int_\IR e^{-i\xi x}{\widehat{g^{C,P}}(\xi-i\eta)}\varphi^{\text{bs}}_{t,\sigma}(-(\xi-i\eta))\d{\xi}.
	\end{split}
	\end{equation}
The integral in~\eqref{eq:bsrhsPsi02} is a Fourier (inversion) integral. We read off
	\begin{equation}
	\label{eq:bsrhsPsi03}
	\begin{split}
		\widehat{\psi_\eta}(t,\xi) 
		=&\ \widehat{g^{C,P}}(\xi-i\eta)\exp\left(-t\left(r+ A^{\text{bs}}(\xi-i\eta)\right)\right),
	\end{split}
	\end{equation}
where we used the relation between the characteristic function and the symbol of a process. Now,
	\begin{equation}
	\label{eq:bsrhsPsi04}
	\begin{split}
		&\widehat{\frac{\partial}{\partial t}\psi_\eta}(t,\xi) 
		= -\left(r+ A^{\text{bs}}(\xi-i\eta)\right)\widehat{\psi_\eta}(t,\xi).
	\end{split}
	\end{equation}
Finally, since $\psi^{\text{bs}}\in H_\eta^{\alpha/2}(\IR)$, we have 
that
	\begin{equation}
	\label{eq:bsrhsPsi05}		
		\int_\IR(\mathcal{A}\psi^{\text{bs}})(t,x)\varphi_j(x)\d{x} = \frac{1}{2\pi}\int_\IR A(\xi-i\eta)\widehat{\psi^{\text{bs}}(t, \cdot)}(\xi-i\eta)\overline{\widehat{{\varphi_j}_{-\eta}}(\xi)}\d{\xi}.
	\end{equation}
Collecting our results proves the claim.
\end{proof}

The candidate $\psi=\psi^{\text{bs}}$ matches the desired criteria. It is quickly evaluable, since functions for yielding Black-Scholes prices are implemented in many code libraries. Also, the integral in~\eqref{eq:FjBlackScholes} is numerically accessible, since the integrand decays fast. Observe that FFT techniques could be employed to computed $F_j(t)$ for all $j=1,\dots,N$ simultaneously. A major disadvantage of choosing $\psi=\psi^{\text{bs}}$, however, lies in the fact that $t\in[0,T]$ can not be separated from the integrand in~\eqref{eq:FjBlackScholes}. Consequently, $F_j(t^k)$, $j=1,\dots,N$, $k=1,\dots,M$, must be numerically evaluated on each time grid node individually. This results in significant numerical cost. We therefore present a second candidate for $\psi$ that avoids this issue.


\begin{lemma}[Subtracting Quasi-Hockey stick multiplied by Normal]
\label{lem:SubtractingQHSnormal}
Let $\sigma_\psi>0$. Define $\psi^C$ in the call option and $\psi^P$ in the put option case by
	\begin{equation}
	\label{eq:cfrhsnorm00.1}
	\begin{alignedat}{2}
		\psi^C(t,x) =&\ \left(e^x-Ke^{-r t}\right)\Phi(x), &(t,x)\in[0,T]\times[a,b],\\
		\psi^P(t,x) =&\ \left(Ke^{-r t}-e^x\right)\left(1-\Phi(x)\right),\qquad &(t,x)\in[0,T]\times[a,b],
	\end{alignedat}
	\end{equation}
where $\Phi$ denotes the cumulative distribution function of the normal $\mathcal{N}(0,\sigma_\psi^2)$ distribution. Furthermore, in the call option case choose $\eta<-1$ and $\eta>0$ in the put option case. Then, the right hand side $F:[0,T]\rightarrow \IR^N$ computes to
	\begin{multline}
	\label{eq:subtractqhsnormal}
		F_j(t) = \frac{1}{2\pi}\Bigg(\int_\IR\big(A(\xi-i\eta)+r\big)\frac{\widehat{f^\mathcal{N}}(\xi-i(\eta+1))}{i\xi+(\eta+1)}\overline{\widehat{{\varphi_j}}(\xi+i\eta)}\d{\xi}\\ - e^{-r t}K\int_\IR\big(A(\xi-i\eta)\big)\frac{\widehat{f^\mathcal{N}}(\xi-i\eta)}{i\xi+\eta}\overline{\widehat{{\varphi_j}}(\xi+i\eta)}\d{\xi}\Bigg),
	\end{multline}
for all $j=1,\dots,N$ with $t\in[0,T]$, where $A$ is the symbol of the associated operator $\mathcal{A}$ and with
	\begin{equation*}
		\widehat{f^\mathcal{N}}(\xi) = \exp\left(-\frac{1}{2}\sigma^2_\psi\xi^2\right),
	\end{equation*}
the Fourier transform of the normal $\mathcal{N}(0,\sigma_\psi^2)$ density.
\end{lemma}

\begin{proof}
We consider the call option case first. To derive the expression for $F_j$ in~\eqref{eq:subtractqhsnormal} we need to compute the Fourier transform of (the appropriately weighted) $\psi^C$. We choose $\eta<-1$ arbitrary but fix and $t\in[0,T]$ arbitrary but fix and compute for $K=1$,
	\begin{equation}
	\label{eq:cfrhsnorm01}
	\begin{split}
		\widehat{\psi_\eta^C(t,\cdot)}(\xi) =&\ \int_\IR e^{i\xi x}e^{\eta x}\left(e^x-e^{-r t}\right)\Phi(x)\d{x}\\
		=&\ \int_\IR e^{i\xi x}e^{(\eta+1) x}\Phi(x)\d x - e^{-r t}\int_\IR e^{i\xi x}e^{\eta x}\Phi(x)\d{x}.
	\end{split}
	\end{equation}
Integration by parts and \lHopital 's rule that
	\begin{equation}
	\label{eq:cfrhsnorm04}
	\begin{split}	
		\int_\IR e^{i\xi x}e^{(\eta+1) x}\Phi(x)\d x =&\ - \frac{1}{i\xi+(\eta+1)}\int_\IR e^{i(\xi-i(\eta+1))x}f^\mathcal{N}(x)\d x,
	\end{split}
	\end{equation}
which can be expressed in terms of the Fourier transform of the normal distribution yielding
	\begin{equation}
	\label{eq:cfrhsnorm05}
	\begin{split}	
		\int_\IR e^{i\xi x}e^{(\eta+1) x}\Phi(x)\d{x} =&\ - \frac{\widehat{f^\mathcal{N}}(\xi-i(\eta+1))}{i\xi+(\eta+1)}.
	\end{split}
	\end{equation}
Equivalently, we obtain for the second integral in~\eqref{eq:cfrhsnorm01} that
	\begin{equation}
	\label{eq:cfrhsnorm06}
	\begin{split}	
		\int_\IR e^{i\xi x}e^{\eta x}\Phi(x)\d x =&\ - \frac{\widehat{f^\mathcal{N}}(\xi-i\eta)}{i\xi+\eta}.
	\end{split}
	\end{equation}
Assembling these results we find
	\begin{equation}
	\label{eq:cfrhsnorm07}
	\begin{split}
		\widehat{\psi_\eta^C(t,\cdot)}(\xi) =&\ -\frac{\widehat{f^\mathcal{N}}(\xi-i(\eta+1))}{i\xi+(\eta+1)} +  e^{-r t}\frac{\widehat{f^\mathcal{N}}(\xi-i\eta)}{i\xi+\eta}.
	\end{split}
	\end{equation}
	We deduce
	\begin{multline}
	\label{eq:cfrhsnorm11}
		F_j(t) = \frac{1}{2\pi}\Bigg(\int_\IR\big(A(\xi-i\eta)+r\big)\frac{\widehat{f^\mathcal{N}}(\xi-i(\eta+1))}{i\xi+(\eta+1)}\overline{\widehat{{\varphi_j}}(\xi+i\eta)}\d{\xi}\\
		- e^{-r t}\int_\IR A(\xi-i\eta)\frac{\widehat{f^\mathcal{N}}(\xi-i\eta)}{i\xi+\eta}\overline{\widehat{{\varphi_j}}(\xi+i\eta)}\d{\xi}\Bigg)
	\end{multline}
with
	\begin{equation*}
		\widehat{f^\mathcal{N}}(\xi) = \exp\left(-\frac{1}{2}\sigma^2_\psi\xi^2\right).
	\end{equation*}
For the put option case we choose as defined in~\eqref{eq:cfrhsnorm00.1},
	\begin{equation}
	\label{eq:cfrhsnorm12}
	\begin{split}
		\psi^P(x, t) =&\ \left(Ke^{-r t}-e^x\right)\left(1-\Phi(x)\right)\\
			=&\ \left(e^x - Ke^{-r t}\right)\left(\Phi(x)-1\right).
	\end{split}
	\end{equation}
Since
	\begin{equation}
		\frac{\partial}{\partial x}\left(\Phi(x)-1\right) = \frac{\partial}{\partial x}\Phi(x),\qquad\forall x\in\IR,
	\end{equation}
the computations for $\widehat{\psi^P_\eta}$ follow along the same lines as they do for the call and we get the relation
	\begin{equation}
		\widehat{\psi^P_{\eta}(t,\cdot)}(\xi) = \widehat{\psi^C_{\eta}(t,\cdot)}(\xi),\qquad \forall (t,\xi)\in [0,T]\times\IR,
	\end{equation}
for $\eta$ set to some $\eta>0$, which proves the claim.
\end{proof}

\begin{remark}[Computational features of $\psi^C$ and $\psi^P$]
While $\psi^C$ serves as localizing function for the call option case, $\psi^P$ can be used in the put option case. Both candidates are based on the payoff functions of call and put options but avoid the lack of differentiability with respect to $x$ in $x=\log(Ke^{-rt})$ for $t\in[0,T]$. As a consequence, both $\psi^C$ and $\psi^P$ are very smooth functions and thus fulfil the requirements collected above when $\sigma_\psi$ is chosen small enough. 
Additionally, the two integrals in~\eqref{eq:subtractqhsnormal} do not depend on the time variable $t\in[0,T]$ and thus need to be computed only once for each basis function $\varphi_j$. This results in a significant acceleration in computational time compared to the suggestion $\psi=\psi^{\text{bs}}$ of Lemma~\ref{lem:SubtractingBSprices}.
\end{remark}

\subsection{Empirical Convergence Results}

The previous sections have outlined the necessary consecutive phases in setting up a Finite Element solver for option pricing. In order to obtain a flexible FEM solver for pricing PIDEs in L\'evy models, the previous section 
introduced the symbol method which considers all components of the FEM solver in Fourier space. There, components are based on the symbol instead of the L{\'e}vy measure and become numerically accessible. Many examples of asset models for which the associated symbols exist in analytically closed form have deemed this alternative approach being worthwhile to pursue. At the same time, however, smoothness of the FEM basis functions became a critical issue. 
In a second step, we therefore investigated two approaches to combine smoothness and numerical accessibility. Mollified hat functions and splines were introduced as promising examples to construct a symbol method based FEM solver with.

This section will put that promise to the test. 
We therefore implemented the symbol method for both mollified hats and splines. 
 In stark contrast to a direct implementation of \eqref{eq-ajump} and \eqref{eq:rhsMerton01}, the symbol method enjoys the flexibility of being able to easily plug in the symbol of any L{\'e}vy model for which it is available in analytically closed form. The model restriction of that first approach thus disappears. Having first implemented the method for the Merton model, the extension of the code to the NIG and the CGMY model comes with virtually no additional effort. In this regard, the method impressively underlines its appeal for applications in practice where the suitability of a model might depend on the asset class it is employed for. An institution that needs to maintain pricing routines for several asset classes will thus cherish the flexibility that the symbol method offers, recall Algorithm~\ref{algo:SymbolMethod} in this regard which sketches the implementation of a general, symbol method based FEM solver that easily adapts to various models.

Finally, we conduct an empirical order of convergence study. We consider the univariate Merton, CGMY and NIG model and investigate the empirical rates of convergence for the different implementations as Table~\ref{tab:EOCModelsBasisFcns} summarizes.
\begin{table}[h!]
\begin{center}
\makebox[0pt]{\begin{tabular}{@{}lcllcc@{}}
\toprule\vspace{1ex}
\multirow{2}{*}{\textbf{Model}} & \multirow{2}{*}{\textbf{Symbol}} & \multicolumn{2}{c}{\multirow{2}{*}{\textbf{Parameter choices}}} & \multicolumn{2}{c}{\textbf{Implemented basis functions}} \\
                    &  & &        & Mollified hats      & Splines      \\
\hline\\
\multirow{2}{*}{Merton}              & \multirow{2}{*}{Example~\ref{ex:defMertonsymbol}} & $\sigma=0.15$, & $\alpha=-0.04$, &   \multirow{2}{*}{\checkmark}    & \multirow{2}{*}{\checkmark}  \\
 & & $\beta = 0.2$, & $\lambda=3$  & & \vspace{2ex}\\
\multirow{2}{*}{CGMY}                & \multirow{2}{*}{Example~\ref{ex:CGMYsymbol}} & $C=0.5$, & $G=23.78$,    &     \multirow{2}{*}{\checkmark}  &  \multirow{2}{*}{\checkmark} \\     
& & $M=27.24$,& $Y=1.1$ & &  \vspace{2ex}\\
\multirow{2}{*}{NIG}                 & \multirow{2}{*}{Example~\ref{ex:NIGsymbol}}  & $\alpha=12.26$, & $\beta=-5.77$, &        \multirow{2}{*}{\checkmark}     & \multirow{2}{*}{\checkmark}  \\
& & $\delta=0.52$ & & & \\
\bottomrule\\
\end{tabular}}
\caption{An overview of the models considered in the empirical order of convergence analysis and their parametrization. For these models, the symbol method is implemented and tested for both mollified hat functions and splines. In addition, we investigate the empirical convergence rate for the Merton model using classic hat functions as basis functions in a classic implementation disregarding the symbol method. In all models, the constant risk-less interest rate has been set to $r=0.03$.}
\label{tab:EOCModelsBasisFcns}
\end{center}
\end{table}

For each model and each implemented basis function type enlisted in Table~\ref{tab:EOCModelsBasisFcns} we conduct an empirical order of convergence study using the pricing problem of a call option with strike $K=1$ as an example, thus considering the payoff function
	\begin{equation}
	\label{eq:EOCPayoff}
		g(x) = \max(e^x-1,0).	
	\end{equation}
In each study we compute FEM prices for $N_k$ basis functions, with
	\begin{equation}
		N_k = 1 + 2^k,\qquad k=4,\dots,9
	\end{equation}
resulting in  $N_4=17$ basis functions in the most coarse and $N_9 = 513$ basis functions in the most granular case. On each grid, the nodes that basis functions are associated with are equidistantly spaced from another and the basis functions always span the space interval $\Omega = [-5,5]$. The time discretization is kept constant with $N_\text{time grid}=2000$ equidistantly spaced time nodes spanning a grid range of two years up until maturity, thus covering a time to maturity interval of
	\begin{equation}
		[T_1,T_{N_\text{time}}],\qquad \text{ with  $T_1=0$ and $T_{N_\text{time}}=2$}.
	\end{equation}
For each $k=4,\dots,9$, the resulting price surface constructed by $N_k$ basis functions in space and $N_\text{time}=2000$ grid points in time is computed. A comparison of these surfaces is drawn to a price surface of most granular structure based on the same type of basis function. We call this most granular surface \emph{true} price surface. It rests on $N_\text{true} = N_{11} = 1+2^{11} = 2049$ basis functions in space and $N_\text{time}$ grid points in time spanning the same grid intervals as above, that is $\Omega = [-5,5]$ in space and $[0,2]$ in time, respectively. The underlying FEM implementation is thus based on distances $h_\text{true}$ between grid nodes that basis function are associated with of
	\begin{equation}
	\begin{split}
		h^\text{mollified hat}_\text{true} =&\ (5-(-5))/(2+2^{11}) \approx 0.0049,\\
		h^\text{splines}_\text{true} =&\ (5-(-5))/(4+2^{11}) \approx 0.0049,\\
		\Delta t_\text{true} =&\ 2/(2000-1) \approx 0.001
	\end{split}
	\end{equation}
in space and time, respectively. Note that all space grids are designed in such a way that the log-strike $\log(K) = 0$ is one of the space nodes. For each model and method and each $k=4,\dots,9$, the (discrete) $L^2$ error $\varepsilon_{L^2}$ is calculated as
	\begin{equation*}
		\varepsilon_{L^2}(k) = \sqrt{\Delta t_\text{true}\cdot h_\text{true}\cdot\sum_{i=1}^{N_\text{time}}\sum_{j=1}^{N_\text{true}}\Big(Price_\text{true}(i,j)-Price_k(i,j)\Big)^2},
	\end{equation*}
wherein $Price_\text{true}(i,j)$ is the value of the true pricing surface at space node $j\in \{1,\dots,1+2^{11}\}$ and time node $i=1,\dots,2000$ and $Price_k(i,j)$ is the respective, linearly interpolated value of the coarser pricing surface with only $N_k$ basis function nodes. 

Figure~\ref{fig:EOCanalysisMHatSpline} summarizes the results of the six studies of empirical order of convergence in the Merton, the NIG and the CGMY model in a symbol based implementation once using mollified hats and once using splines as basis functions.
\begin{figure}[t]
\begin{center}
\makebox[0pt]{\includegraphics[scale=1]{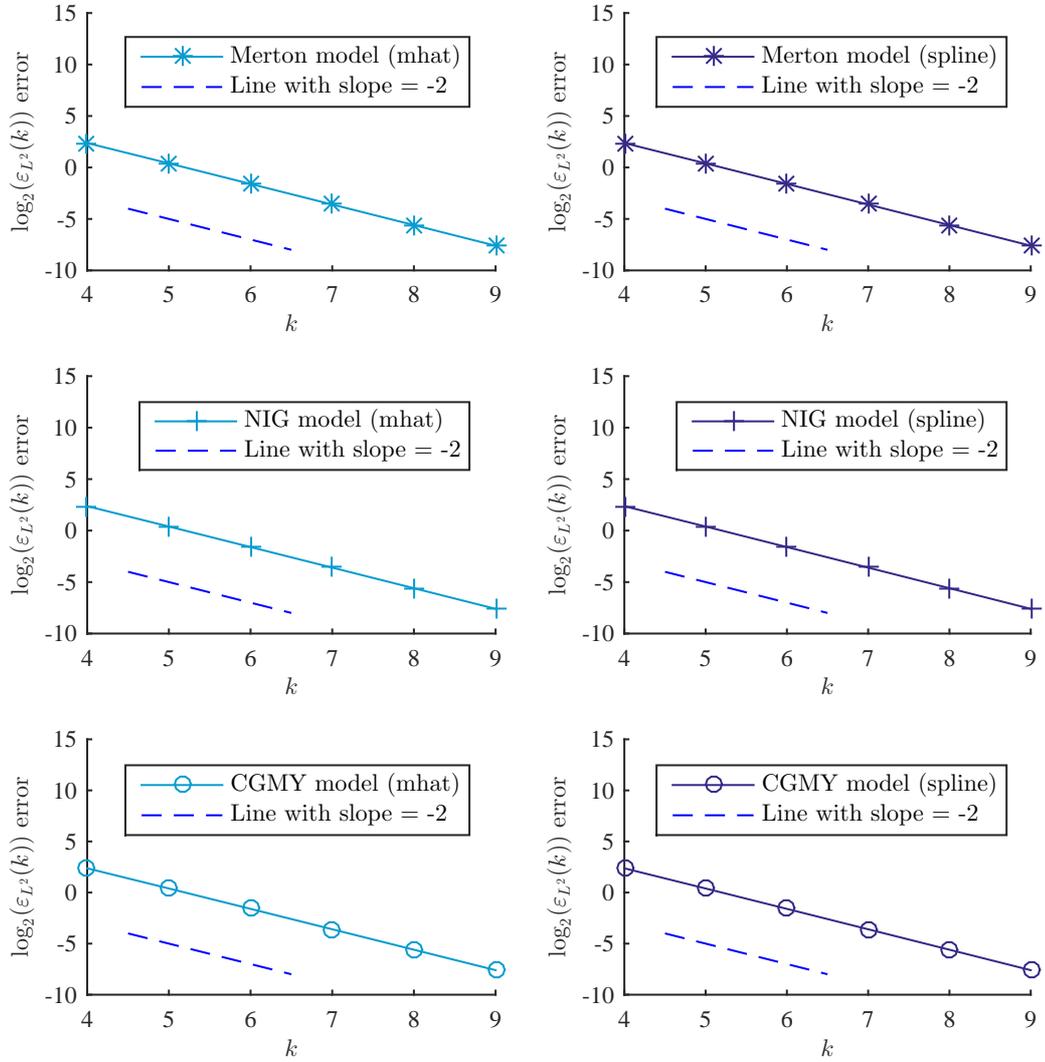}}
\caption{Results of the empirical order of convergence study for the Merton, the NIG and the CGMY model using mollified hats (left pictures) and splines (right pictures) as basis functions. All models are parametrized as stated in Table~\ref{tab:EOCModelsBasisFcns}. Additionally, part of a straight line with (absolute) slope of $2$ is depicted in each figure serving as a comparison.}
\label{fig:EOCanalysisMHatSpline}
\end{center}
\end{figure}
In each implementation and for all considered models, the (discrete) $L^2$ error decays exponentially with rate $2$. 
The convergence result of Theorem 5.4 by \cite{SchwabWavelet03} suggest that this is the best possible rate we can hope for, which yields the experimental validation of both approaches.
\FloatBarrier

\section{Conclusion and outlook}\label{sec-outlook}

We have presented a finite element solver that is highly flexible in the model choice and that maintains numerical feasibility. Invoking the symbol was key. The transition into Fourier space has introduced smoothness as a new requirement to the basis functions. We have presented splines and mollified hats as compatible basis functions in our approach. Several numerical examples have confirmed the convergences rates expected by standard theory in both cases. The contribution of the mollification to the error has not yet been addressed in the literature. One possibility to derive such error estimates lies in adapting the perturbation analysis of~\cite{SchwabWavelet03} treating mollification as a numerical perturbation of the stiffness matrix. Let us mention several possible extensions of the approach. First, the implementation naturally extends to time-inhomogeneous L{\'e}vy models that we neglected here for notational convenience. Second, combining the symbol method with Wavelet basis functions allows for compression techniques that might further improve the overall numerical performance, as \cite{HilberReichmannSchwabWinter2013} outlines. Third, the polynomial decay that we observe in our numerical experiments can possibly be improved to exponential rates by invoking an $hp$-discontinuous Galerkin scheme, see e.g.~\cite{SchoetzauSchwab2006}.



\appendix


\section{Proofs}
\subsection{Proof of Lemma \ref{lem-a=forH1} }\label{sec-proof-lem-a=}
\begin{proof}
We first consider $\varphi,\psi\in C^\infty_0(\rr)$.

For $F\equiv0$ the assertion follows directly from partial integration. Since the L\'evy measure may be unbounded around the origin, the representation of the  jump part of the bilinear form,
\begin{align*}
a^{jump}(\varphi,\psi ) 
:= 
&-\int_{\rr}\int_{\rr}\Big( \varphi(x+y)-\varphi(x) -  \varphi'(x) \,  h(y)\Big)F(\dd y)\psi(x) \ee{2\skl\eta,x\skr}\dd x,
\end{align*}
needs to be carefully derived.
In order to exploit the identity
\begin{align*}
\varphi(x+y) - \varphi(x) - y\varphi'(x) = \int_0^y\int_0^z \varphi''(v)\dd v \dd z
\end{align*}
we split the integral with respect to the L\'evy measure in  three parts, set $c(F):= \int_{|y|<1}\big( y-h(y)\big)F(\dd y) - \int_{|y|>1}h(y)F(\dd y)$ and obtain
\begin{align*}
a^{jump}(\varphi,\psi ) 
:= 
&-\int_{\rr}\int_{|y|<1} \int_0^y\int_0^z \varphi''(x+v)\dd v \dd z F(\dd y)\psi(x) \ee{2\skl\eta,x\skr}\dd x\\
 &- c(F)\int_{\rr}\varphi'(x)\psi(x) \ee{2\skl\eta,x\skr}\dd x
\\
&-\int_{\rr}\int_{|y|>1}\big( \varphi(x+y)-\varphi(x)\big)F(\dd y)\psi(x) \ee{2\skl\eta,x\skr}\dd x.
\end{align*}
Thanks to $\int_0^y\int_0^z \big|\varphi''(v)\big|\dd v \dd z\le c y^2$ with some constant $c>0$ for all $y\in[-1,1]$and
\begin{align}\label{absch-ajump-cont-H1}
\begin{split}
\lefteqn{
\int_{\rr}\int_{|y|<1} \int_0^y\int_0^z\big|\varphi'(x+v)\big|\dd v \dd zF(\dd y)\big|\psi'(x) + 2\eta \psi(x) \big|\ee{2\skl\eta,x\skr}\dd x
}\qquad\qquad\qquad\qquad\qquad\qquad\\
&\le
(1+2\eta)\|\varphi\|_{H^1_\eta} \|\psi\|_{H^1_\eta}\int_{|y|<1}y^2 F(\dd y)
\end{split}
\end{align}

 we can apply the theorem of Fubini and partial integration to obtain
\begin{align*}
\lefteqn{-\int_{\rr}\int_{|y|<1} \int_0^y\int_0^z \varphi''(x+v)\dd v\dd z F(\dd y)\psi(x) \ee{2\skl\eta,x\skr}\dd x}\qquad\qquad\\
&=
\int_{\rr}\int_{|y|<1} \int_0^y\int_0^z \varphi'(x+v)\dd vF(\dd y)\big(\psi'(x) + 2\eta \psi(x) \big)\ee{2\skl\eta,x\skr}\dd x.
\end{align*}
This yields the assertion for $\varphi,\psi\in C^\infty_0(\rr)$.

Next, we verify that the bilinear form as stated in Lemma \ref{lem-a=forH1} is well defined for $\varphi,\psi\in H^1_\eta(\rr)$ and is continuous with respect to the norm of $H^1_\eta(\rr)$. For $F\equiv0$ this is obvious. The assertion follows for the jump part from inequality \eqref{absch-ajump-cont-H1}
and 
\begin{align*}
\int_{\rr}\int_{|y|>1}\big| \varphi(x+y)-\varphi(x)\big|F(\dd y)\big|\psi(x)\big| \ee{2\skl\eta,x\skr}\dd x
&\le 2 F\big(\rr\setminus[-1,1]\big)\|\varphi\|_{L^2_\eta} \|\psi\|_{L^2_\eta}.
\end{align*}
Thus $a$ from Lemma \ref{lem-a=forH1} is a continuous bilinear form on $H^1_\eta(\rr)\times H^1_\eta(\rr)$ that coincides with \eqref{def-aCinfty0-eta} on the dense subset $C^\infty_0(\rr)\times C^\infty_0(\rr)$. This proves the assertion.
\end{proof}

\subsection{Proof of Lemma \ref{lem-conv1}}\label{sec-proof-lem-conv1}

\begin{proof}
To prove the assertion, we verify the conditions of  Lemma 3 in \cite{Glau2016}, which provides an abstract robustness result for weak solutions. We first observe that the conditions for $f_n,f,g_n,g$ coincide with those of Lemma 3 in \cite{Glau2016}. Second, we verify conditions (An1)--(An3) of Lemma 3 in \cite{Glau2016}. Therefore we assign to each $u,v\in X$ the coefficients $\alpha_k(u),\alpha_k(v)\in \rr$ for $k\le N$ such that $u=\sum_{k=1}^N\alpha_k(u)\varphi_k$ and  $v=\sum_{k=1}^N\alpha_k(v)\varphi_k$. Thanks to the finite dimensionsionality of $X$, there exists a constant $\widetilde{C}>0$ such that for all $u\in X$,
\begin{equation}\label{eq-normequiv}
\|u\|_V \le \sum_{k=1}^N\big|\alpha_k(u)\big|\|u\|_V\le C' \|u\|_V .
\end{equation}
Thanks to \eqref{cond-an} there exists a sequence $0<c_n\to 0$ such hat for all $j,k\le N$,
\begin{align}
\big|(a_n-a)(\varphi_j,\varphi_k)\big|
&\le c_n\|\varphi_j\|_V\|\varphi_k\|_{V}.\label{cond-an2}
\end{align}
Together with assumption (A2) this yields for all $j,k\le N$,
\begin{align}\label{cond-an1}
\big|a_n(\varphi_j,\varphi_k)\big|
&\le C_1\|\varphi_k\|_{V}\|\varphi_k\|_{V}.
\end{align}
Inequalities \eqref{cond-an1} and \eqref{eq-normequiv} together yield for all $u,v\in X$,
\begin{align*}
\big|a_n(u,v)\big|
&\le 
\sum_{k=1}^N \sum_{j=1}^N\big|\alpha_k(u)\alpha_j(v)\big| \big|a_n(\varphi_k,\varphi_j)\big|\\
&\le
C_1\sum_{k=1}^N \sum_{j=1}^N\big|\alpha_k(u)\alpha_j(u)\big| \|\varphi_k\|_{V}\|\varphi_k\|_{V}\\
&\le
C_1\widetilde{C}^2 \|u\|_V \|v\|_V,
\end{align*}
which shows that condition (An1) of Lemma 3 in \cite{Glau2016} is satisfied. 

Due to inequalities 
\eqref{cond-an2} and \eqref{eq-normequiv}, we have for all $u\in X$,
\begin{align*}
\big|(a-a_n)(u,u)\big|
&\le
\sum_{k=1}^N \sum_{j=1}^N\big|\alpha_k(u)\alpha_j(u)\big| \big|a_n(\varphi_k,\varphi_j)\big|\\
&\le 
c_n \sum_{k=1}^N \sum_{j=1}^N\big|\alpha_k(u) \alpha_j(u)\big| \|\varphi_j\|_V\|\varphi_k\|_{V}\\
&
\le c_n \widetilde{C}^2 \|u\|_V^2,
\end{align*}
which shows assumption (An3) of Lemma 3 in \cite{Glau2016}.

Finally, from assumption (A1) and the last inequality for all $u\in X$ we obtain
\begin{align*}
a_n(u,u)
&\ge
a(u,u) - \big|(a-a_n)(u,u)\big|\\ 
&\ge
G\|u\|_V^2- G'\|u\|_H^2
- c_n \widetilde{C}^2 \|u\|_V^2,
\end{align*}
which shows that there exists $N_0\in \nn$ such that condition (An2) of Lemma 3 in \cite{Glau2016} is satisfied for all $n>N_0$. This shows the assertion of Lemma \ref{lem-conv1}.
\end{proof}

\vspace{4ex}

\bibliographystyle{elsarticle-harv}
  \bibliography{LiteraturFKac}

\end{document}